\documentclass[final,3p,times,twocolumn,authoryear]{elsarticle}

\usepackage{algorithm}
\usepackage{algpseudocode}
\usepackage{pgfplots}
\usepackage{mathtools}
\usepackage{nicematrix}
\usepackage{multicol}
\usepackage{graphicx}
\usepackage[normalem]{ulem}
\usepackage{svg}
\newcommand{\real}{\mathbb{R}}
\newcommand{\Tau}{\mathcal{T}}
\usepackage{amssymb}
\pgfplotsset{compat=1.18}

\journal{Applied Soft Computing}

\begin{document}

\begin{frontmatter}



\title{Monte Carlo Optimization for Solving
Multilevel Stackelberg Games}


\author[inst1]{Pravesh Koirala}

\affiliation[inst1]{organization={Vanderbilt University},
city={Nashville},
country={USA}}

\author[inst1]{Forrest Laine}

\begin{abstract}
Stackelberg games originate where there are market leaders and followers, and the actions of leaders influence the behavior of the followers. Mathematical modelling of such games results in what's called a Bilevel Optimization problem. There is an entire area of research dedicated to analyzing and solving Bilevel Optimization problems which are often complex, and finding solutions for such problems is known to be NP-Hard. A generalization of Stackelberg games is a Multilevel Stackelberg game where we may have nested leaders and followers, such that a follower is, in turn, a leader for all lower-level players. These problems are much more difficult to solve, and existing solution approaches typically require extensive cooperation between the players (which generally can't be assumed) or make restrictive assumptions about the structure of the problem. In this paper, we present a stochastic algorithm to approximate the local equilibrium solutions for these Multilevel games. We then construct a few examples of such Multilevel problems, including: a) a nested toll-setting problem; and b) an adversarial initial condition determination problem for Robust Trajectory Optimization. We test our algorithm on our constructed problems as well as some trilevel problems from the literature, and show that it is able to approximate the optimum solutions for these problems within a reasonable error margin. We also provide an asymptotic proof for the convergence of the algorithm and empirically analyze its accuracy and convergence speed for different parameters. Lastly, we compare it with existing solution strategies from the literature and demonstrate that it outperforms them.
\end{abstract}


\end{frontmatter}

\newcommand{\olsi}[1]{\,\overline{{#1}}} 

\section{Introduction}
\label{sec:introduction}
Stackelberg Equilibriums are well-known and extensively studied economic phenomena. In their most rudimentary form, they occur when there is a market \textit{leader} whose decision influences one or many market \textit{followers}. These leaders and followers are constrained in their own way and are assumed to be rational players who seek to minimize their costs (or maximize their profits) while satisfying their constraints. Mathematical modeling of these games gives rise to a Bilevel Optimization problem of the following form:

\begin{align*}
\min_{x_1\in \real^{n_1}, x_2 \in \real^{n_2}} ~~&f^1(x_1, x_2) \\
s.t. ~~&g^1(x_1, x_2) \ge 0\\
&x_2 \in \arg\min_{x_2 \in \real^{n_2}} ~~f^2(x_1, x_2) \\
&~~~~~~~~~~~~~~~~~s.t.~~g^2(x_1, x_2)\ge 0
\end{align*}
Where the upper level player with the objective $f^1(x_1, x_2): \real^{n_1+n_2} \mapsto \real$ and constraints $g^1(x_1, x_2): \real^{n_1+n_2}\mapsto \real^{m_1}$ optimizes over $x_1 \in \real^{n_1}$ knowing that its choice of $x_1$ causes the lower-level player to adapt its response variable $x_2\in \real^{n_2}$ to minimize its objective $f^2(x_1, x_2): \real^{n_1+n_2}\mapsto \real$ subject to its constraints $g^2(x_1, x_2):\real^{n_1+n_2}\mapsto\real^{m_2}$. It is also generally assumed that the objective functions $f^1$ and $f^2$ and the constraints $g^1$ and $g^2$ are twice differentiable. But even with these assumptions, solution set for problems of this form generate not only non-convex but also non-smooth manifolds. In fact, finding an equilibrium point or a solution to these problems is known to be NP-Hard \citep{ben1990computational, blair1992computational}. Popular strategies to solve these problems are Vertex Enumeration methods \citep{bialas1984two}, Complementary Pivoting methods \citep{judice1992sequential}, Mixed Integer Programming or Branch and Bound methods \citep{bard1990branch}, and meta-heuristics based methods such as Genetic Algorithms \citep{oduguwa2002bi} and Particle Swarm optimization \citep{Han2016AST} etc.

The problem discussed above is called a Bilevel problem because it has two levels of optimizers (alternatively referred to as decision makers or players in this text) with their own sets of decision variables and constraints. A natural extension of such a leader-follower game is, then, a Multilevel Stackelberg game that can be modeled as:
\begin{align*}
&Level_1& &~~~~~~~~~~~\hdots\\
&& &~~~~~~~~~~~~~\vdots\\
&Level_l & \min_{x_l, x_{l+1}..., x_L} ~~&f^l(X) \\
&& s.t. ~~&g^l(X) \ge 0\\
&Level_{l+1} & &x_{l+1},...,x_L \in \arg\min_{x_{l+1},...,x_L} ~~f^{(l+1)}(X) \\
&& &~~~~~~~~~~~~~~~~~~~~s.t.~~g^{(l+1)}(X)\ge 0\\
&& &~~~~~~~~~~~~~\vdots\\
&Level_{L}& &~~~~~~~~~~~\hdots\\
\end{align*}
Where $l \in {1, 2, ..., L}$ indicates player level and $x_l \in \real^{n_l}$ is the variable that player $l$ controls. Similarly, $X \in \real^{n}$ (where $n=n_1+n_2+...+n_L$) is the concatenation of all $x_l$'s and, therefore, the entire search space of the problem. The objectives and constraints for each player are defined as $f^l: \real^{n}\mapsto\real$ and $g^l:\real^{n}\mapsto\real^{m_l}$. It is often assumed that no two players share any degrees of freedom (or decision variables) but we make no such assumptions in this work. To be precise, any player at level $l$ is a Stackelberg follower of all preceding players at level $1, 2, ...~ l-1$ and is simultaneously a Stackelberg leader for all players at level $l+1~...~L$. Like before, each player cares for their own objective and has their own constraints. To define the solution of the problem, we start with the concept of a rational reaction set for the final player L, $\phi^L(x_1, ... x_{L-1})$ defined as:
\begin{align*}
\phi^L(x_1, ... x_{L-1}) &:= \arg\min_{x_L} f^L(X) \\
&~~~~s.t.~~g^L(X) \ge 0
\end{align*}
Then, the rational reaction set for any player $l$, i.e. $\phi^l(x_1...x_{l-1})$ can be recursively defined as:
\begin{align*}
\phi^l(x_1...x_{l-1}) &:= \arg\min_{x_l, x_{l+1},~...~x_L} f^l(X) \\
&~~~~~~~~~s.t.~~g^l(X) \ge 0 \\
&~~~~~~~~~~(x_{l+1}...x_L) \in \phi^{l+1}(x_1...x_l)
\end{align*}
The solution to the entire problem is then: \begin{align*}
\phi^1 &:= \arg\min_{x_1,~...~x_L} f^1(X) \\
&~~~~s.t.~~g^1(X) \ge 0 \\
&~~~~~~~~~~(x_{2}...x_L) \in \phi^2(x_1)
\end{align*}
It must be noted that in general, $\phi^l$ may not be a singleton, and therefore, there can be multiple local solutions for the problem.

These problems are not new and have been researched over the years in domains such as Economics, Optimal Control, Operations Research, and Decision Programming etc., for example, to model multi-stakeholder fund allocation, supply chain networks, inventory management, and power system security \citep{han2015tri, 4252262, Fard2018ATL, Cassidy1971EfficientDO}. There are further generalizations of multilevel problems that include multiple players at each level (sometimes called a Multilevel Decentralized problem) who have equal deciding power amongst themselves but are, as before, followers for players above them and leaders for players below them. In this work, we restrict ourselves to multilevel optimization problems with only a single decision maker at each level and introduce a monte-carlo sampling based method to find solutions of such multilevel optimization problems. We also model a robust trajectory optimization problem and a generalized version of the toll-setting problem as multilevel problems and use our algorithm to find solutions for them. In summary, the main contributions of this paper are:
\begin{itemize}
\item A simple yet powerful monte-carlo method to solve multilevel problems.
\item Modeling adversarial initial condition determination and nested toll-setting problem as multilevel optimization problems and obtaining their solutions via our proposed algorithm.
\end{itemize}
The remainder of this paper is structured as follows: In section \ref{section:lit}, we explore some of the works related to such multilevel optimization problems, including some of the algorithms proposed to solve them. In section \ref{section:mcmo}, we propose a stochastic algorithm to solve problems of this kind. Then, in section \ref{section:problems}, we construct two such multilevel problems: a) a robust optimization problem of finding adversarial initial condition, and b) a nested toll-setting problem, and discuss the nature of their solutions. Then, in section \ref{sec:experiment}, we apply this algorithm to solve a few problems from existing literature in addition to the constructed problems from section \ref{section:problems} and compare the obtained solutions. In section \ref{sec:comparisons}, we perform empirical comparisons to study the convergence speed and computation time of the proposed algorithm. Finally, in section \ref{section:conclusions} we pave the way for further research by outlining some of the possible improvements we envision in this domain and proceed to conclude the work with a brief recap.

\section{Literature Review}
\label{section:lit}

Stackelberg games and Bilevel Optimizations are well-researched problems, and we refer readers to \citet{Dempe2020BilevelOT} in lieu of attempting a survey ourselves. Henceforth, we limit ourselves to works related to trilevel or general multilevel problems.
\subsection{Linear Multilevel Problems}\cite{Cassidy1971EfficientDO} first modeled the flow of resources between the federal, state, and municipal levels as a trilevel problem and provided a recursive dynamic algorithm for solving such problems. \cite{Bard1984AnIO} later established stationarity conditions for trilevel linear optimization problems, generalized it to p-level stationarity problems, and devised a Cutting plane algorithm to solve them. \cite{ue1986hybrid} devised a hybrid method based on the K-th best algorithm and Parametric Complementary Pivot algorithm to solve trilevel linear problems. \cite{anandalingam1988mathematical} devised another method for solving trilevel linear problems by first obtaining and embedding the first-order necessary conditions (FONCs) of the third-level problem into the second-level problem, then obtained FONCs of thusly obtained problem and embedded it into the first-level problem. \cite{Benson1989OnTS} investigated a specific case where linear multilevel problems are unconstrained and performed rigorous geometric analysis. Their major result was to show that the feasible solution set of such problems is a union of connected polyhedral regions. \citet{White1997PenaltyFA} modified \citet{Bard1984AnIO} 's method by changing the first step in their algorithm and claimed a qualitative improvement on the overall results.
\subsection{Fuzzy Set / Goal Programming-Based Approaches}
\cite{Lai1996HierarchicalOA} considered a fuzzy set based algorithm to model and solve linear bilevel and multilevel problems. \cite{Shih1996FuzzyAF} later improved it to model problems that are not just hierarchical but also decentralized, or both, in nature. \cite{Pramanik2007FuzzyGP} modeled the multilevel problem as a fuzzy goal programming problem to solve it. \cite{Zhang2010ModelSC} presented a kth-best algorithm to solve linear trilevel programming problems and solved a constructed problem of annual budget allocation in a company with CEO, branch heads, and group supervisors.

\subsection{Meta-heuristics based approaches}
\cite{woldemariam2015systematic} developed a genetic algorithm based method to solve arbitrarily deep multilevel problems for bounded decision variables. \cite{Han2016AST} devised a particle swarm optimization based method to solve bilevel problems and used it to solve a trilevel problem as well by embedding the stationarity conditions of the last level problem into the second level problem and converting the entire structure into a bilevel programming problem. At this point, we must also mention \cite{Lu2016MultilevelDA}'s survey of multilevel decision-making problems, which, although a bit dated, is an excellent resource for multilevel problems, algorithms, and applications developed until 2016.
\subsection{Applications}
\cite{Han2017TrilevelDF} used Vertex Enumeration method to solve a decentralized supply chain network involving manufacturers, logistic companies, and consumers modeled as a trilevel decentralized programming problem. \cite{Fard2018ATL} modeled a multi-stakeholder supply chain problem as a trilevel problem and used five different meta-heuristic algorithms to solve them by solving each level in a turn-based fashion. They also later modeled a tire closed-loop supply chain network as a trilevel problem and solved it using a similar approach \citep{Fard2018HybridOT}. \cite{tilahun2012new} developed a turn-based optimization strategy similar to \cite{Fard2018ATL} to solve general Multilevel problems and later generalized it to solve fuzzy Multilevel, multi-objective problems with collaboration. \cite{Tilahun2019FeasibilityRA}. \cite{Tian2019MultilevelPC} formulated a coordinated cyber-attack scenario as a trilevel problem and used the column and constraint generation method to obtain a solution. \cite{luo2020energy} modeled an Energy scheduling problem as a trilevel optimization problem and exploited its structure to obtain a closed analytical expression. \cite{laine2023computation} later developed a general algorithm to find solutions to Generalized Feedback Nash Equilibrium problems, which can be modeled as a Multilevel Stackelberg problem.

From the literature review, it is clear that multiple methods exist to solve trilevel problems, but only a few of these can be generalized to solve an arbitrarily deep multilevel problem. Even then, we find that each method has its own limitations. For instance, fuzzy set based methods (\cite{Lai1996HierarchicalOA, Shih1996FuzzyAF, Pramanik2007FuzzyGP}) implicitly assume some degree of cooperation from lower levels, which is not an assumption that holds for every problem. Similarly, turn-based methods of \cite{tilahun2012new, Tilahun2019FeasibilityRA, Fard2018ATL, Fard2018HybridOT} are iterative best response algorithms that are more suited to find solutions to Nash equilibrium problems, and since they do not take into account the rational reactions of lower-level players, they do not converge towards the Stackelberg equilibrium. \cite{woldemariam2015systematic}'s genetic algorithm is quite promising, but it only works for bounded variables, which makes it inapplicable for a wide class of problems. Similarly, \cite{laine2023computation}'s algorithm is applicable only under assumptions of strong complementarity.

In light of these facts, we propose an algorithm in section \ref{section:mcmo} that solves all of the outlined concerns above. Furthermore, we demonstrate in section \ref{sec:comparisons} that even though it's simple and intuitive, it outperforms the existing methods of similar nature. Compared to other algorithms, our proposed algorithm has the advantage that:
\begin{itemize}
\item It can handle problems with unbounded decision variables and, thus, is applicable to a wider class of problems.
\item It can handle problems with non-differentiable objectives, as long as the final objective is differentiable.
\item It can handle equality constraints present at the final level, unlike other Meta-heuristic algorithms, which fail to handle any equality constraints at all without any reformulations.
\item It does not require any reformulations of the objective functions and, thus, can solve problems that can't be approached via KKT or Value function based reformulations.
\item It's an anytime algorithm and can be tuned to obtain arbitrary accuracy at the expense of computation.
\end{itemize}

\section{Monte Carlo Multilevel Optimization (MCMO)}
\label{section:mcmo}

Some of the notations used in the algorithm are as follows:
\begin{align*}
L \in \mathbb{N} &: \text{Number of players}.\\
n_l \in \mathbb{N} &: \text{Number of variables for player }l\\
x_l \in \real^{n_l} &: \text{Variables that $l$-th player controls}\\
C^l &: \text{Feasible region for player }l\\
X \in \real^{n} &: \text{Concatenation of all $x_l$'s}\\
\end{align*}
\begin{align*}
C = \bigcap_{l=1}^{L} C^l &: \text{Feasible region for the problem }\\
x_s &: \text{Initially feasible point s.t. } x_s \in C \\
f^l &: \text{Objective function of player }l ~(\mathbb{R}^{n}\mapsto \mathbb{R}) \\
D^l := \real^{n_l} &: \text{Subspace spanned by}
~x_l.\\
\alpha^l \in \real^+ &: \text{Step size for player } l\\
N^l \in \mathbb{N} &: \text{Number of samples generated for player }l\\
M^l \in \mathbb{N} &: \text{Number of sampling iterations for } l
\end{align*}

Apart from the notations above, we use some colloquial array notations as follows:
\begin{align*}
[~] &: \text{Empty array}\\
X[a:b] &: \text{Slice of X from index a to b inclusive}\\
X[a:end] &: \text{Slice of X from index a to the length}\\
&~~~ \text{of X inclusive} \\
X~.+b &: \text{A broadcasting summation operator.}\\
\end{align*}

MCMO is a sampling based algorithm. It iteratively refines any given approximate solution by generating samples in its neighborhood. These samples are successively passed down to each lower-level players, who generate samples of their own and pass them down to their lower-level players. This continues until the very last level, where a solver is used to obtain the solution for $x_L$ given the variables $x_1, ... x_{L-1}$ for the corresponding objective and constraints. Once these solutions are obtained, they are returned to upper-level players who evaluate them, select the best among them for their own objectives, and subsequently return them to their upper levels. At level 1, all returned solutions are evaluated, and the best among them is kept as the current estimate of the solution. In this way, MCMO acts as a gradient-free solver and does not require gradient information for any objective or constraint function except the last one. Similarly, since the last level is always solved by using a solver, MCMO can accommodate both equality / inequality constraints for that level so long as it's supported by the solver. MCMO is described in Algorithm \ref{algo:mcmo}. In essence, it takes an initially feasible point and continuously searches in its neighborhood for a better feasible and optimal point for a specified number of iterations. When the desired number of iterations is reached, MCMO returns a smoothed result from the last $k$ obtained iterates, as outlined in subsection \ref{subsec:smoothing}.

\begin{algorithm}
\caption{$MCMO~(x_s, k)$}
\begin{algorithmic}[1]
\State{$X \gets x_s$}
\State{$P \gets [X]$}
\For{$i \in [1 ... maxiter]$}
\State {$X \gets OPTIMIZE(X, 1) ~or~ X$ \textit{\# stick with same point if no better point found.}}
\State{$P \gets P \cup X$}
\EndFor
\State{$\textbf{return} SMOOTHEN(P, k)$}
\end{algorithmic}
\label{algo:mcmo}
\end{algorithm}

The Optimize function defined in Algorithm \ref{algo:optimize} takes as input an initially feasible point $x_s$ and a level $l$ ($=1$ for initial call). For the final player ($l=L$), this function uses a solver, IPopt \citep{wachter2006implementation} in this case, to optimize for the final objective $f^L$ subject to the constraints $C^L$. In all other cases, it generates $N^l+1$ random directions (including the \textbf{zero} direction) in the subspace $D^l$ to obtain new candidate points, which are then recursively passed to the optimizers of the lower-level player, i.e., $l+1$. These passed candidate points are then recursively perturbed by the lower-level players and returned. Out of all the returned values, player $l$ keeps the perturbed candidates that's best for its objective and satisfies its feasibility constraints. This process is repeated by the player $l$ for $M^l$ number of times, where, at the end of each such sampling iteration, it chooses the point that is the best among all obtained candidates in that iteration and uses it for the next iteration. At the end, it returns the final obtained best candidate to the upper player $l-1$. In the event that no feasible point can be found at any iteration, the last known best candidate point is retained and used for the next iteration. If no feasible point can be found even after $M^l$ iterations, the function returns $null$. In this way, this function can obtain solution for multilevel problems with arbitrary levels.

The algorithm uses three sub-procedures \textit{SOLVE\_FULL}, \textit{ARGMIN}, and \textit{RAND\_DIRECTIONS}. These sub-procedures are intuitive, and thus, we only explain but do not explicitly outline them here. \textit{SOLVE\_FULL} takes, in order, an initial point, an objective, a set of constraints, and player level (to determine degrees of freedom to optimize on) and uses a solver to fully solve it to completion. Similarly, \textit{ARGMIN} takes, in order, a list of candidate points, the player level $l$, and determines the best point according to the objective function $f^l$ ignoring any null points in the given list. Finally, \textit{RAND\_DIRECTIONS} generates $N^l$ random directions from a uniform hypercube (of length 1) centered at the origin in the subspace $D^l$.

\begin{algorithm}
\caption{$OPTIMIZE~(X, l)$}
\begin{algorithmic}[1]
\Statex{\textit{Require: $f^l, C^l, M^l, \alpha^l, N^l, D^l$}}
\State{$X_R \gets null$}
\If{$l = L$}
\State{$X_R \gets SOLVE\_FULL(X, f^L, C^L, L)$}
\If{$X_R \not\in C^L$}
\State{\textbf{return} \textit{null}}
\Else
\State{\textbf{return} $X_R$}
\EndIf
\EndIf
\For {$k \in \{1, ..., M^l\}$}
\State{\textit{\# generate candidate points}}
\State{$X_C \gets X ~{.+}~ \{ \alpha^l \cdot $}
\Statex{$~~~~~~~~~~~~~~~~~~~~~~~~~~~~~RAND\_DIRECTIONS(N^l, D^l) \cup \textbf{0} \}$}
\For{$x \in X_C$}
\State{$x \gets OPTIMIZE(x, l+1)$}
\If{$x \not\in C^l$}
\State{\textbf{continue}}
\EndIf
\State{$X_R \gets ARGMIN(X_R, x, l)$}
\EndFor
\State{$X \gets X_R ~or~ X$}
\EndFor
\State{\textbf{return} $X_R$}
\end{algorithmic}
\label{algo:optimize}
\end{algorithm}

Ideally, MCMO should be used with a high number of samples and sampling iterations, i.e. $N^l, M^l$ to obtain accurate results, as only by doing so can we solve all lower levels to completion before optimizing any upper-level problem. But this can result in a lot of computational overhead, as outlined in section \ref{sec:comparisons}. So, in practice, we select a reasonable $N^l, M^l$ for each level while still keeping the problem computationally tractable. But this will result in stochastic estimates (as opposed to true solutions), which is precisely what limits MCMO to an approximate algorithm.

\subsection{Initialization}
\label{subsec:initialization}
MCMO requires that an initially feasible (not necessarily optimal) point $x_s \in C$ be provided. A viable option to achieve such an initially feasible point is to solve the following problem:
\begin{align*}
x_s = &\arg\min_{X}~~ 0\\
&~~~~~~~~ s.t. ~X \in C
\end{align*}
However, a heuristic to achieve a reasonably optimal starting point for a non-trivial problem is to take the weighted sum of the objectives. Which is to say, we solve the following optimization problem to obtain such an initially feasible point:
\begin{align*}
x_s = &\arg\min_{X}~~ \sum_{l=1}^{L} w_l f^l(X)\\
&~~~~~~~~~~~ s.t. ~X \in C
\end{align*}
Where, $w_l$'s are chosen as required. This heuristic yields better starting points in cases where the true solution lies closer to the pareto front of the involved objective functions.

\subsection{Smoothing}
\label{subsec:smoothing}
Since MCMO is a stochastic algorithm, it can only provide approximate solutions as all the lower-level problems are not completely solved. This is especially true when the number of samples generated ($N^l$) or the number of sampling iterations ($M^l$) are too low and $\alpha^l$ is high. Therefore, at the end of the algorithm, the last $k<maxiter$ points are used to obtain a more stable approximation of the equilibrium point by using a smoothing scheme. The choice of smoothing scheme may depend upon the problem, but in this work, we use the following scheme:

\noindent \textbf{Best Objective Smoothing Scheme:} $X^*$ is approximated as $X^*=\arg\min_{X \in \{X_1, X_2, ... X_k\}} f^1 (X)$. Where $f^1$ is the objective function of the first player. This scheme is guaranteed to produce a feasible point (since all $X_1, X_2, ... X_k$ are feasible, as shown in subsection \ref{subsec:feasible}).

\subsection{Practical Considerations}
\label{subsec:pracon}
The performance of the algorithm is reliant on the number of samples $N^l$ generated per level, number of sampling iterations $M^l$, choice of $\alpha^l$, and the number of iterations $maxiter$. In general, more samples and sampling iterations would improve the accuracy of the solution, but at the expense of computation costs. Similarly, a large $\alpha^l$ may prevent convergence, whereas a low $\alpha^l$ would delay it. An appropriate way of running MCMO is thus to start off with low $N^l, M^l$ and high $\alpha^l$ and then fine-tune the result with a lower value of $\alpha^l$ and a higher sample size $N^l$ and iteration $M^l$ to the desired accuracy. Due to the nature of multilevel optimization problems, lower-level players must be provided with greater deciding powers than any upper-level players. This is especially true when the degrees of freedom are shared between the upper and lower level players. Consider for example, the problem 
\begin{alignat*}{2}
    &\max_{x} ~x \\
             &~~s.t. ~x \in &&\arg\min_x x\\
             &&&s.t.~x \in [l, u]
\end{alignat*}
In this case, the solution for this problem for $x\in[l, u]$ is $l$. In terms of games, when the first-level player chooses any $x=x'$, the second-level player will choose $x=l$, overriding any choice of the variable $x$ made by the upper-level player.
Therefore, to achieve true solutions, $\alpha^l,$ $N^l, M^l$ for each subsequent level should be increased. Additionally, the choice of $N^l$ should also consider the degrees of freedom. If a player has control of two variables, they must be allowed to sample more directions than if they only had one decision variable. This ensures that the sampling is fair for all levels.

However, for simple problems, it may also be desirable to use the same $\alpha$ per player for maintaining a lean parameter space. And if bounds on the player's variables are known, it can guide the choice of $\alpha$.

\subsection{Computation Time}
MCMO is a recursive sampling based algorithm and thus, its computation time increases exponentially with each additional level. Furthermore, the computation time will also depend upon the parameters $N^l, M^l, maxiter$ and the nature of the problem itself. While parallelizing the implementation may provide speedups, for this work, we do not attempt such efforts and have left it for future improvements. A detailed empirical analysis of computation time can be found in section \ref{sec:comparisons}.

An implementation of the algorithm can be found on GitHub (https://github.com/VAMPIR-Lab/MCMO).

\subsection{Proofs}
This section presents proofs for the feasibility and convergence of the MCMO algorithm.
\subsubsection{Proof of Feasibility}
\label{subsec:feasible}
\newtheorem{lemma}{Lemma}
\newproof{proof}{Proof}

\begin{lemma}
\label{lemma:inductionbase}
Any non-null point X returned from a function call of the form $X=~$OPTIMIZE($\cdot$ , l) is feasible for level l, i.e., $X \in C^l$.
\end{lemma}

\begin{proof}
For the final level $l=L$, this is easy to see from Algorithm \ref{algo:optimize} lines 4–8. If a point is infeasible for that level, the if condition on line 4 causes a \textit{null} return. Otherwise, a feasible point is returned in line 7. For levels $l\neq L$, a non-null result can only be returned if $X_R$, which is initially null, is set with a non-null value $x$ in line 18. But line 18 can only execute if the feasibility condition of line 15 was satisfied, which means that the returned non-null value $X\in C^l$.
\end{proof}
\begin{lemma}
\label{lemma:inductionhypothesis}
Any non-null point X returned from a function call of the form $X=$ $OPTIMIZE(\cdot, l)$ is always obtained from a lower-level function call of the form $OPTIMIZE(\cdot, l+1)$ for $l < L$.
\end{lemma}
\begin{proof}
Since $l \neq L$, following arguments similar to lemma \ref{lemma:inductionbase}, it must have been set by line 18. But any such point is clearly obtained in line 14 by function call of the form $OPTIMIZE(\cdot, l+1)$. Hence, this is true.
\end{proof}
We can now prove the following claim:

\newtheorem{claim}{Claim}
\begin{claim}
Each iteration in MCMO function obtains a feasible point.
\end{claim}
\begin{proof}
From lemma \ref{lemma:inductionhypothesis}, we know that any non-null point obtained from function of the form $OPTIMIZE(\cdot, 1)$ is obtained from $OPTIMIZE(\cdot, 2)$, $OPTIMIZE(\cdot, 3)$, and so on until $OPTIMIZE(\cdot, L)$. Similarly, we also know from lemma \ref{lemma:inductionbase} that any non-null point thus obtained must be feasible for levels 1, 2, ... $L$-1, and L. Therefore, any non-null point obtained from an iteration of the MCMO algorithm is feasible for all levels. Furthermore, if a null point is obtained at any point, MCMO retains the last non-null point, which is either $x_s$, an initially feasible point, or another non-null point previously obtained in iteration that has already been shown to be feasible.
\end{proof}

\newcommand{\nbd}{\mathcal{N}}

\subsubsection{Proof of Convergence}
Any analytical reasoning for general multilevel problem is decidedly hard, and for stochastic or meta-heuristic algorithms, the difficulty only increases. Thus, we only present an asymptotic proof of convergence for a narrow class of problems that satisfy the following simplifying assumptions:
\begin{enumerate}
\item The rational reaction set $\phi^l(x_1, ..., x_{l-1})$ (as defined in section \ref{sec:introduction}) for player $l$ is a point-to-point map, i.e., all rational reactions are unique for given upper-level decisions.
\item A solution exists for the given problem, and the solver used for the final level can always find solutions when they exist.
\end{enumerate}
In general, assumption 1 may not be valid but may hold if the upper-level constraints are restrictive enough or if the topmost objective is strongly convex and we want to solve an optimistic multilevel optimization problem, i.e. lower levels cooperate with the topmost player for ambiguous rational reactions. Furthermore, this is a simplification that multiple analytic treatments of this problem \citep{liu1998stackelberg, woldemariam2015systematic} have made as arguing about the problem in general is intractable.

Under our assumption, for any multilevel Stackeblerg problem, the optimization that player $l$ solves, say $P^l(x_1, ... x_{l-1})$, condenses to:
\begin{align*}
P^l(x_1, ... x_{l-1}) := \min_{x_l} &f^l(x_1, ..., x_l, \phi^{l+1}(x_1, ..., x_l)) \\
&~~~~s.t.~~g^l(X) \ge 0 \\
\end{align*}

\begin{lemma}
\label{convbase}
$OPTIMIZE(\cdot~, L)$ solves $P^L(x_1, ... x_{L-1})$.
\end{lemma}
\begin{proof}
For the last level, i.e., $l=L$, this function uses a solver to obtain the solution. Since it's assumed that a solution exists and that the solver can find it, this is trivially true.
\end{proof}

\begin{lemma}
\label{convhyp}
If $OPTIMIZE(\cdot~, l+1)$ solves $P^{l+1}(x_1, ... x_l)$, then $OPTIMIZE(\cdot~, l)$ solves $P^l(x_1, ... x_{l-1})$ given $N^l, M^l \to \infty$
\end{lemma}
\begin{proof}
Since infinite samples are assumed with infinite sampling iterations and it's also assumed that a unique reaction (say $x_l^*$) exists for $P^l$, we claim that the sampling process would eventually converge towards $x_l^*$. To show that this is indeed true, we first assume that the sampling does not converge towards the optimum $x_l^*$. This can only mean one of the following:
\begin{enumerate}
\item The algorithm cycles between points $x_l^1, x_l^2, ... x_l^i$. But this must mean that $v(x_l^1) > v(x_l^2) > ... > v(x_l^i) > v(x_l^1)$, which is a contradiction. Here, we define $v(x_l):= f^l(x_1, ..., x_l, \phi^{l+1}(x_1, ... x_l))$.
\item The algorithm gets stuck on some $x_l$ and no $x_l'$ exists in its neighborhood such that $v(x_l')<v(x_l)$ and $x_l'$ satisfies appropriate constraints. However, this, by definition, is a local optimum for the player $l$ and thus, by our assumption, is the same as $x_l^*$ and results in contradiction.

\end{enumerate}

\end{proof}

\begin{claim}
MCMO eventually converges upon the unique solution.
\end{claim}
\begin{proof}
Under our framework, the overall problem reduces to $P^1$. From lemma \ref{convbase} and \ref{convhyp}, we have a proof by induction that MCMO solves $P^1$ when $\forall l, N^l, M^l \to \infty$ by calling $OPTIMIZE(\cdot~, 1)$
\end{proof}

\section{Some Multilevel Problems}
\label{section:problems}
\subsection{Adversarial Initial Condition (AIC) determination problem}
\label{subsection:robust}
We can loosely define a Trajectory as a continuous path in some space. In robotics and control, such paths are generally produced from some initial conditions (start point, environment, etc.) by a set of rules or functions, usually called a \textit{policy}.
This problem is related to finding a worst-case initial condition for any given policy. The worst-case being an initial point from where, if a trajectory is generated according to such a policy, it ends up a) bringing the trajectory as close to touching the obstacle as possible, and b) increasing the length cost of the trajectory. Figure \ref{fig:aic} depicts the problem we construct here.

We consider a 2D plane to be our environment. The blue circular region is the feasible region $\chi \subset \real^2$ where any start point $x \in \real^2$ is allowed to reside. A fixed and known policy $\Pi$ then generates a trajectory $\Tau = \Pi(x) = \tau^0, \tau^1, ...~, \tau^i \in \real^2, ~ \tau^0 = x$ up to the finishing line $D \in \real$ using the start point such that some cost $f(\Tau) \in \real$ (modeled here as the horizontal length of the trajectory i.e. $f(\Tau) = D-\tau^0_1$) is minimized and certain feasibility conditions for each trajectory points are satisfied i.e. $g(\tau^i) \ge 0~\forall \tau^i \in \Tau$.

In this example, the condition of feasibility for a trajectory $\Tau$ is that all trajectory points be outside the obstacle region $\mathcal{O} \subset \real^2$. Modeling $\mathcal{O}$ as a circle centered at $o$ with radius $r$, our feasibility condition for each trajectory point becomes: $g(\tau^i) = ||o-\tau^i||^2 - r^2 \ge 0 ~~ \forall \tau^i \in \Tau$. The problem that we consider in this work is to find an adversarial initial point $x^a$ such that for any given policy $\Pi$, the generated trajectory $\Tau^a = \Pi(x^a)$ is as close to infeasibility and sub-optimality as possible. The rationale being that, with such obtained point, we could iterate our policy to improve it under even the most adverse initial conditions. \textit{We do not attempt policy training in this text and have left it for future work.}

\begin{figure}
\includegraphics[width=0.4\textwidth]{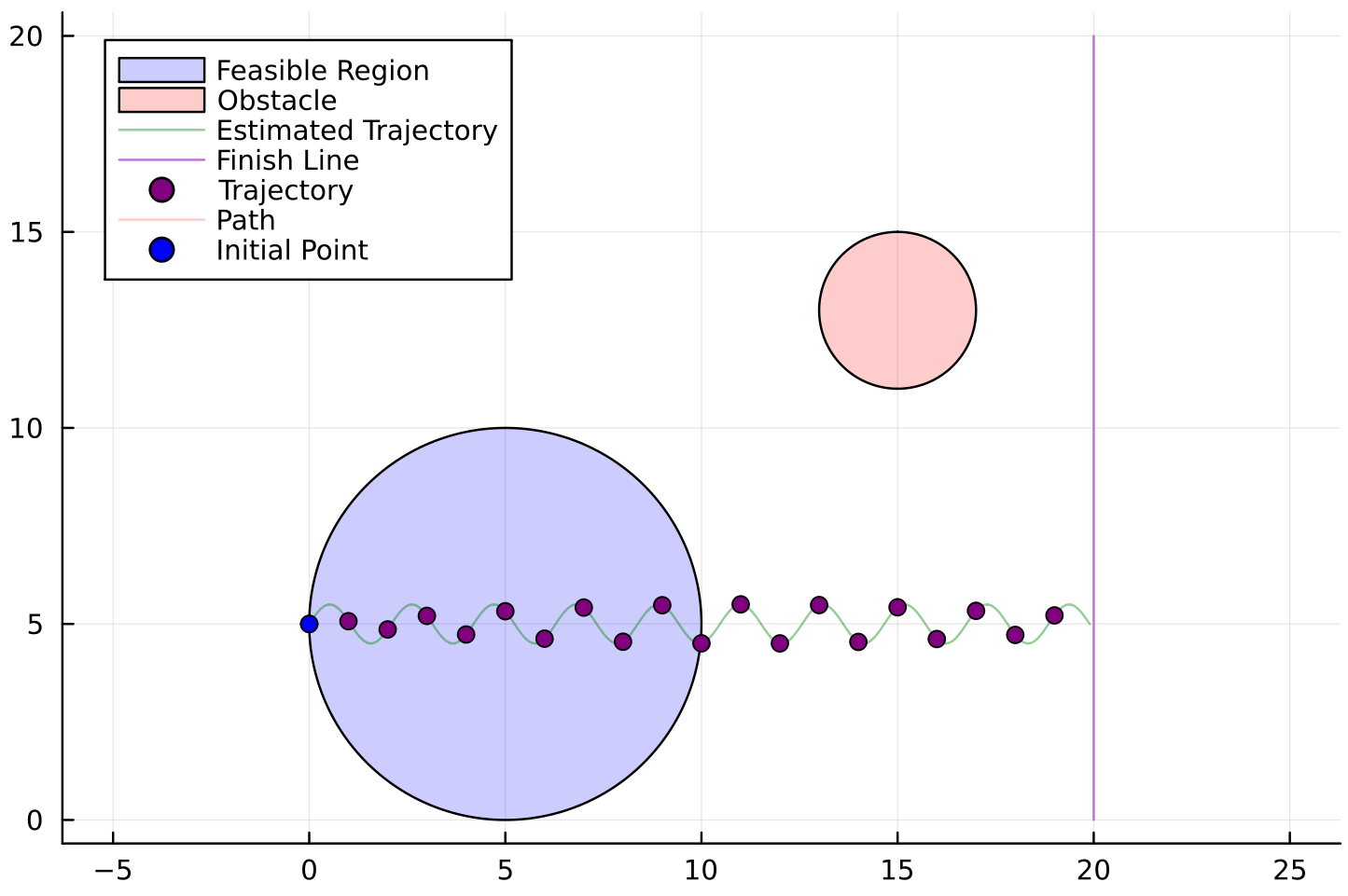}
\caption{The environment for the AIC problem. $\Tau$ is a sampled sinusoidal trajectory.}
\label{fig:aic}
\end{figure}

We can model this problem as a trilevel game as follows:
\begin{subequations}
\begin{align}
\label{eq:trilevel}
&\max_{x,T\in\real}~~f(\Tau) \\
&~~~~ s.t~~x,T\in\arg\min_{x, T}~ T\\
&~~~~~~~~~~~~~~~~~~~ s.t.~~ x \in \chi \\
&~~~~~~~~~~~~~~~~~~T \in \arg\max_{T} ~ T\\
&~~~~~~~~~~~~~~~~~~~~~~~~~~~~~s.t~T \ge 0\\
&~~~~~~~~~~~~~~~~~~~~~~~~~~~~~ \Tau = \Pi(x)\\
&~~~~~~~~~~~~~~~~~~~~~~~~~~g(\tau^i) \ge T;~~ \forall \tau^i \in \Tau
\end{align}
\end{subequations}

An interpretation of the trilevel problem is as follows: The first player wants the initial point $x\in \chi$ for the trajectory to maximize the cost of the trajectory $f(\Tau)$. The second player, whereas, wants to bring $T$ close to 0 by manipulating $x$. But, $T$ is the minimum of all feasibility scores $g(\tau^i)$ i.e., the point closest to violation. So when $T\to 0$, the closest trajectory point to the obstacle becomes even closer, and as a result, the trajectory $\Tau$ touches the obstacle $\mathcal{O}$. Here, the first and second players share the same degree of freedom, i.e., the variable $x$. Generally, in multilevel games such as these, non-overlapping d.o.f's are considered. But, as we mentioned previously, we make no such assumptions and design a general algorithm that can handle all such scenarios.

\subsection{Nested Toll-Setting problem}
A toll-setting problem is a well-known bilevel optimization problem where a toll-setter decides a toll amount for a road segment. Since they want to maximize the total income, they can neither set the toll too high, or drivers will avoid the road segment due to exorbitant fees, nor set the toll too low, or their total income will decrease. We refer readers to \cite{labbe1998bilevel} for a more detailed treatment of this problem. Instead, we focus on a generalization of this problem, i.e., the Nested Toll-Setting problem, as shown in figure \ref{fig:toll}.

\begin{figure}
\centering
\includegraphics[width=0.4\textwidth]{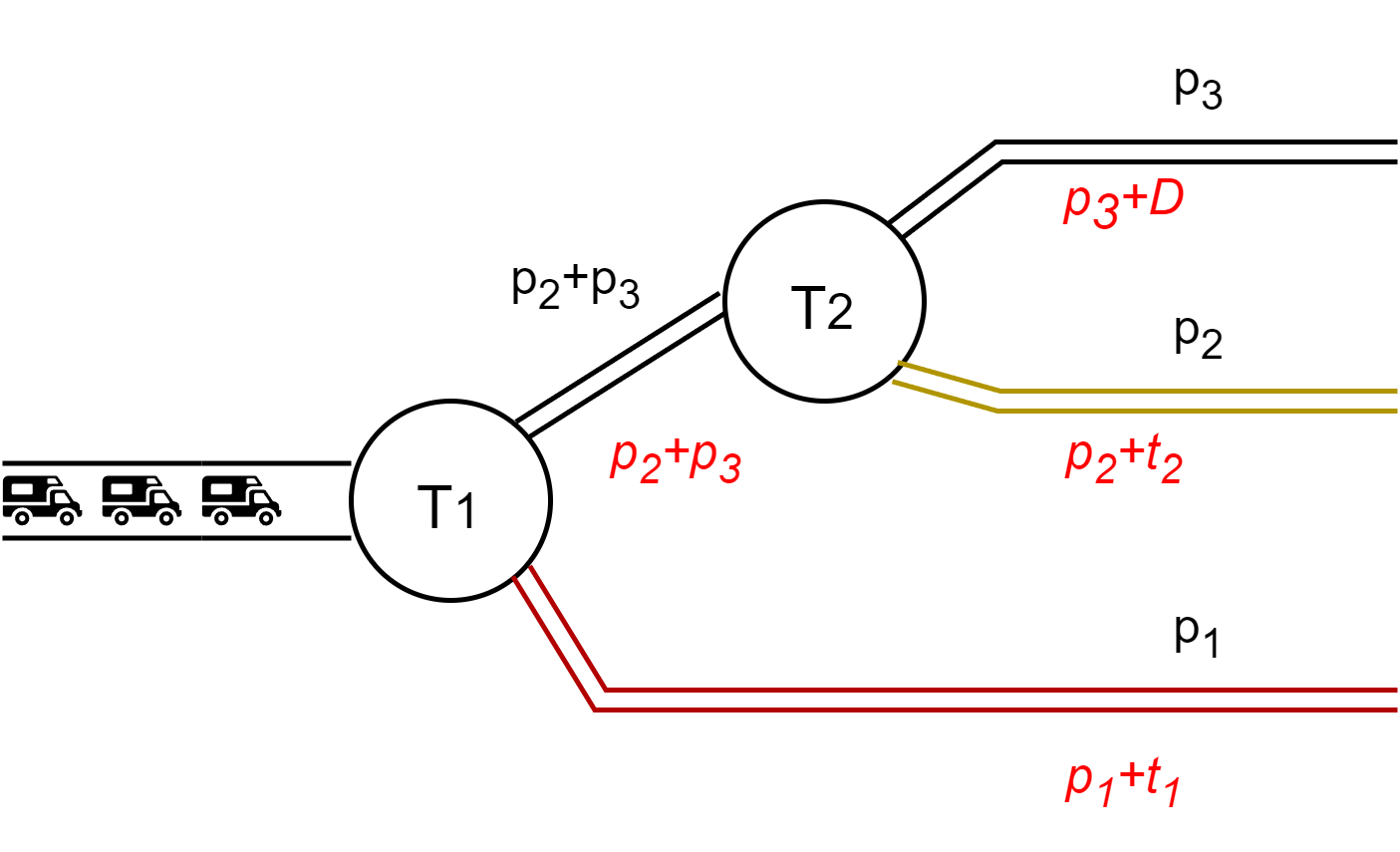}
\caption{The nested toll-setting problem. The numbers in black above a road segment represent the percentage of traffic on that segment. The numbers in red below a road segment represent the cost of taking that segment. Costs are sum of toll price, congestion, and other extra factors.}
\label{fig:toll}
\end{figure}

We consider two toll stations $T_1$ and $T_2$ established to oversee their respective tolled segments (red for $T_1$, and yellow for $T_2$). Any vehicle arriving at $T_1$ has the option to either take the tolled segment (red) by paying $t_1$ cost per unit traffic or take the non-tolled segment (black) for free. We assume that $p_1$ percentage of the original fleet takes the red tolled segment. Similarly, any vehicle arriving at $T_2$ has the option to either take the tolled segment (yellow) by paying $t_2$ per unit traffic or take the non-tolled segment. We assume that $p_2$ percentage of the original fleet takes the yellow tolled segment and $p_3$ percentage of original fleet takes the final free segment. It must be clarified that $p_1, p_2,$ and $p_3$ represent the percentage of fleet that first arrives at $T_1$ establishing $p_1 + p_2 + p_3 = 1$.

From the perspective of the fleet, the cost of travelling through any segment is the sum of a) the toll on the segment, b) congestion on the segment, and c) additional costs associated with the segment. For the purpose of this problem, we establish the congestion cost for any segment to be $\sigma \cdot p$, where $p$ is the percentage traffic on the segment and $\sigma$ is a constant. This is to say that the congestion cost increases linearly with the traffic on the segment. We further simplify this problem by setting $\sigma=1$, thereby setting the congestion cost at each segment equal to the traffic percentage at that segment. Finally, we assume that none of the road segments have any additional costs except for the final free segment, which has an extra cost $D$. This extra cost could theoretically model road length, road conditions, traffic lights, or a myriad of other factors. For this problem, we allow $D$ to be less than 0, allowing it to model a reward or a subsidy as well. Figure \ref{fig:toll} shows the costs associated with each road segment below the segments in red.

Once the toll has been set for any segment, the fleet decides to divide a certain percentage of its traffic to the tolled road or the free road to minimize its total cost across each road segment. Now if, at each toll station, the fleet makes a greedy decision, i.e., deciding whether or not to take the tolled road without considering any future toll stations on the way, then this problem can be written as the following trilevel problem:
\begin{subequations}
\begin{alignat}{2}
\max_{t_1, t_2,p_1,p_2,p_3} & p_1 \cdot t_1 + p_2 \cdot t_2 \\
s.t.~t_1, t_2 &\ge 0 \\
p_1,p_2,p_3&\in\arg\min_{p_1,p_2,p_3} p_1 \cdot (p_1 + t_1) \\
&~~~~~~~~~~~~~~~~~~~~+ (p_2 + p_3)^2\\
&~~~~~~~~~~ s.t.~p_1 \in [0, 1] \\
&~~~~~~~~~~ p_2,p_3\in\arg\min_{p_2, p_3} ~ p_2 \cdot (p_2+t_2) \\
&~~~~~~~~~~~~~~~~~~~~~~~~~~~~~~~~~~+ p_3 \cdot (p_3 + D)\\
&~~~~~~~~~~~~~~~~~~~~~~~~~~~~s.t.~p_2 \in [0,1]\\
&~~~~~~~~~~~~~~~~~~~~~~~~~~~~p_3 \in [0,1]\\
&~~~~~~~~~~~~~~~~~~~~~~~~~~~~p_1+p_2+p_3 = 1
\end{alignat}
\end{subequations}
Here, the first level corresponds to the toll-setter who decides on $t_1, t_2$ to maximize their total income. The second level is the fleet's decision at station $T_1$. It chooses the percentage of traffic to balance total congestion costs and toll costs. The third level is the remaining fleet's decision at station $T_2$ for the same.

\subsubsection{Nature of the solution}
It can be argued with relative ease that for a very high $D$, it's beneficial for the toll-setter to redirect all traffic to $T_2$ whereas for a very low $D$, the toll-setter is better off exacting all tolls from $T_1$ instead. In fact, there are two known equilibrium points $X = (t_1, t_2, p_1, p_2, p_3)$ for this problem for different values of $D$ (see Appendix \ref{appendixA}):
\begin{itemize}
\item At $D = \olsi D = 6, \olsi X = (\ge 2, 4, 0, 1, 0)$
\item At $D = \underline D = -1.5, \underline X = (1, \ge 0, 0.25, 0, 0.75)$
\end{itemize}

\section{Experiments and Results}
\label{sec:experiment}
In this section we solve some existing multilevel problems from the literature in addition to our constructed problems, i.e., the adversarial initial condition (AIC) problem and the nested toll-setting problem outlined in section \ref{section:problems} using MCMO algorithm of section \ref{section:mcmo}. To keep our parameter space restricted and the experiments simple, we run each of these problems with the same value of $\alpha$ for different levels. Furthermore, we also set all $M^l=1$ and instead setup our algorithm based solely on $\alpha^l$ and the number of samples $N^l$. All examples have been run on a personal computer with an Intel Core i5 8400 processor with a 2.8GHz frequency and 32 GBs of DDR4 RAM.

\subsection{Solving AIC using MCMO}
We now solve the AIC problem for two policies, as described below.

\subsubsection{Linear Policy: $\Pi^l$}
We define the linear policy $\Pi^l: \real^2 \mapsto \real^2$ as follows:
$$\Pi^l([x_1, x_2]^T) = [x_1+\delta, x_2]^T$$

Where $\delta \in \real$ is a step-size. Intuitively, this policy takes a point $x^i$ and generates a point $x^{i+1}$ by stepping $\delta$ distance in the $x_1$ axis while leaving $x_2$ unchanged, i.e., a horizontal line parallel to the $x_1$ axis.

\subsubsection{Non-Linear Policy: $\Pi^n$}
We define $\Pi^n : \real^2 \mapsto \real^2$ as follows:
\begin{align*}
\Pi^n([x_1, x_2]^T) = [ &x_1+\delta, \\&x_2 + A \left( \sin( B (x_1 + \delta)) - \sin( B (x_1)) \right)]^T
\end{align*}

Where $\delta \in \real$ is a step-size, $A \in \real$ is an amplitude parameter, and $B \in \real$ is frequency parameter. This generates a sinusoidal trajectory parallel to the $x_1$ axis.

\subsubsection{Setup for AIC}
For both of the trajectories, we apply MCMO to obtain adversarial points for different placements of the obstacle circles of radius r = 2. For all experiments, our feasible region is a circle centered at $[5, 5]^T$ with a radius of 5 r. The number of trajectory points is fixed at $N^\tau = 20$, and the destination plane is set to $D = 20$. For both policies, step-size $\delta$ is set to 1, and for non-linear trajectory $\Pi^n$, $A, B$ are set to $0.5, 3$, respectively. We apply MCMO for a maximum of 150 and use the best objective smoothing scheme with 10 final samples. Similarly, the step parameter "alpha is set to 3. For both policies, $N^1$ was chosen to be $2$ and $N^2$ was chosen to be $10$. In general, $N^2>N^1$ is in accordance with subsection (\ref{subsec:pracon}), but in addition, player 2 has more degrees of freedom as compared to player 1, and furthermore, both player 1 and player 2 share two degrees of freedom $(x_1, x_2)$, so no matter what player 1 chooses, it is modified by player 2, so player 1 has very little influence to begin with. As discussed previously, we set $M^1=M^2=1$. While initializing, we found that the weights $w_1 = 10^5, w_2 = 10^{-5}, w_3 = 1$ (\ref{subsec:initialization}) gave us feasible starts. In general, this will always depend on the problem being solved. The initial points produced for linear policy are shown in figure \ref{fig:linear}, and those for nonlinear policy are shown in figure \ref{fig:nonlinear}.
\begin{figure*}
\centering
\begin{multicols}{2}
\includegraphics[width=0.75\linewidth]{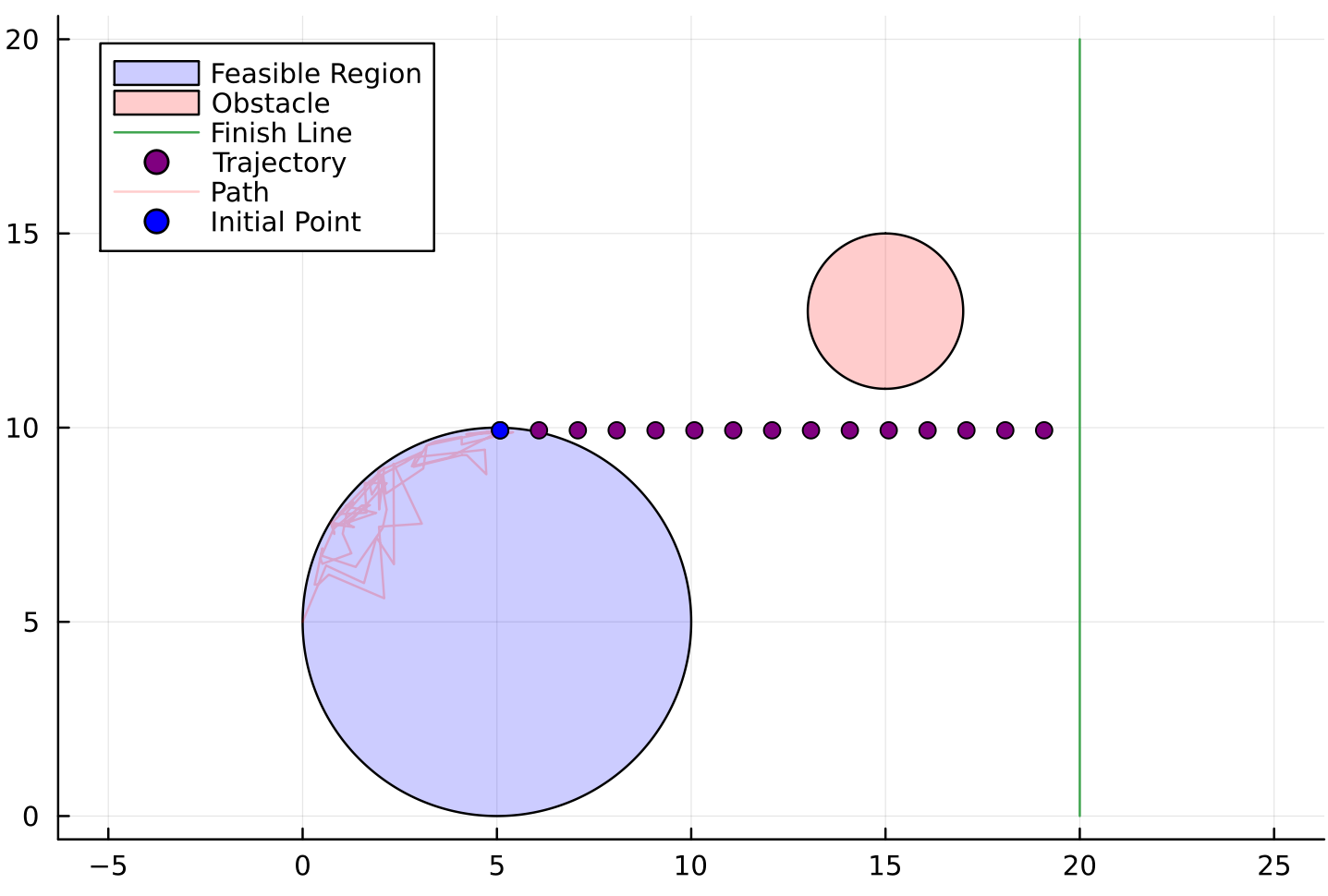}\par
\includegraphics[width=0.75\linewidth]{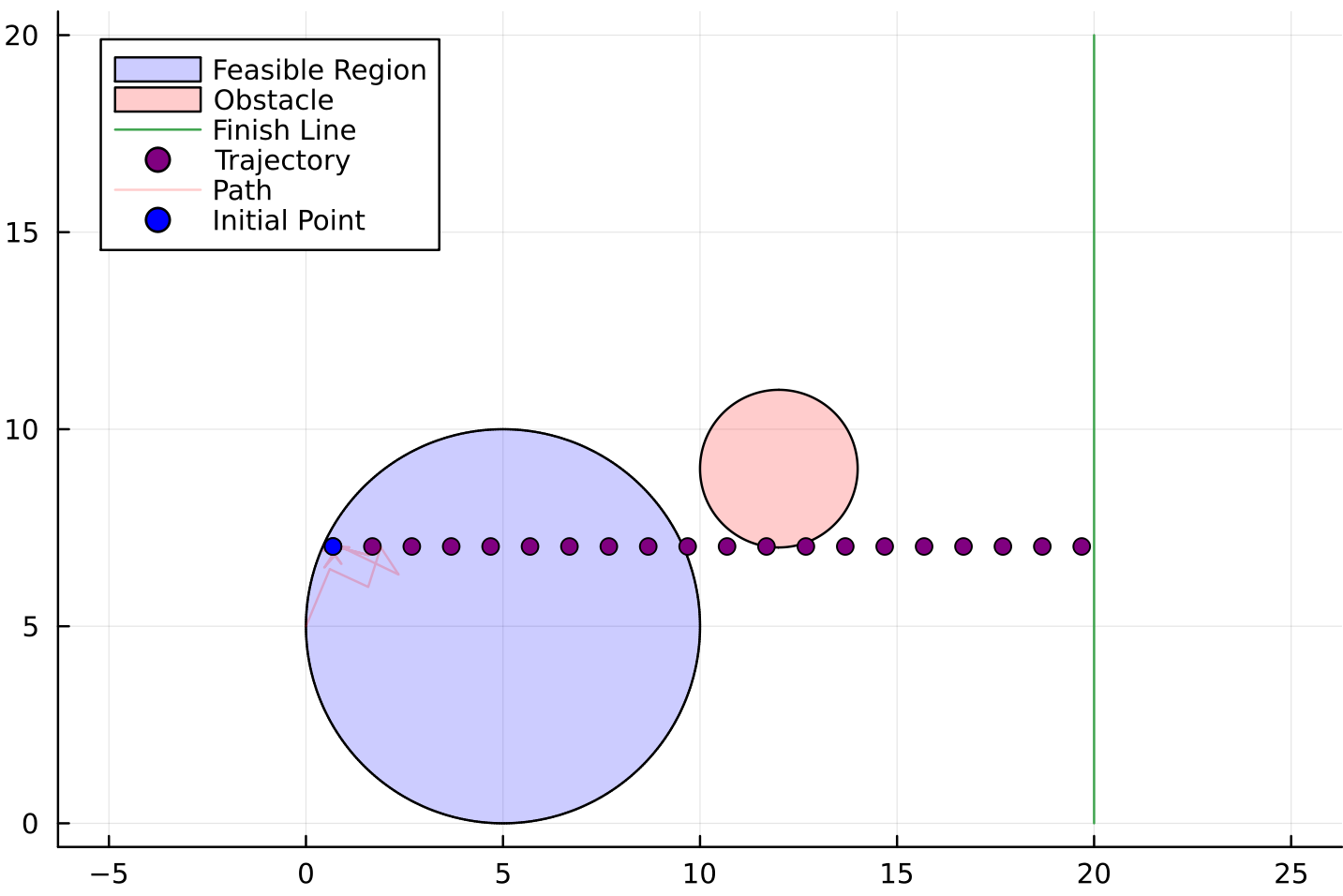}\par
\end{multicols}
\begin{multicols}{2}
\includegraphics[width=0.75\linewidth]{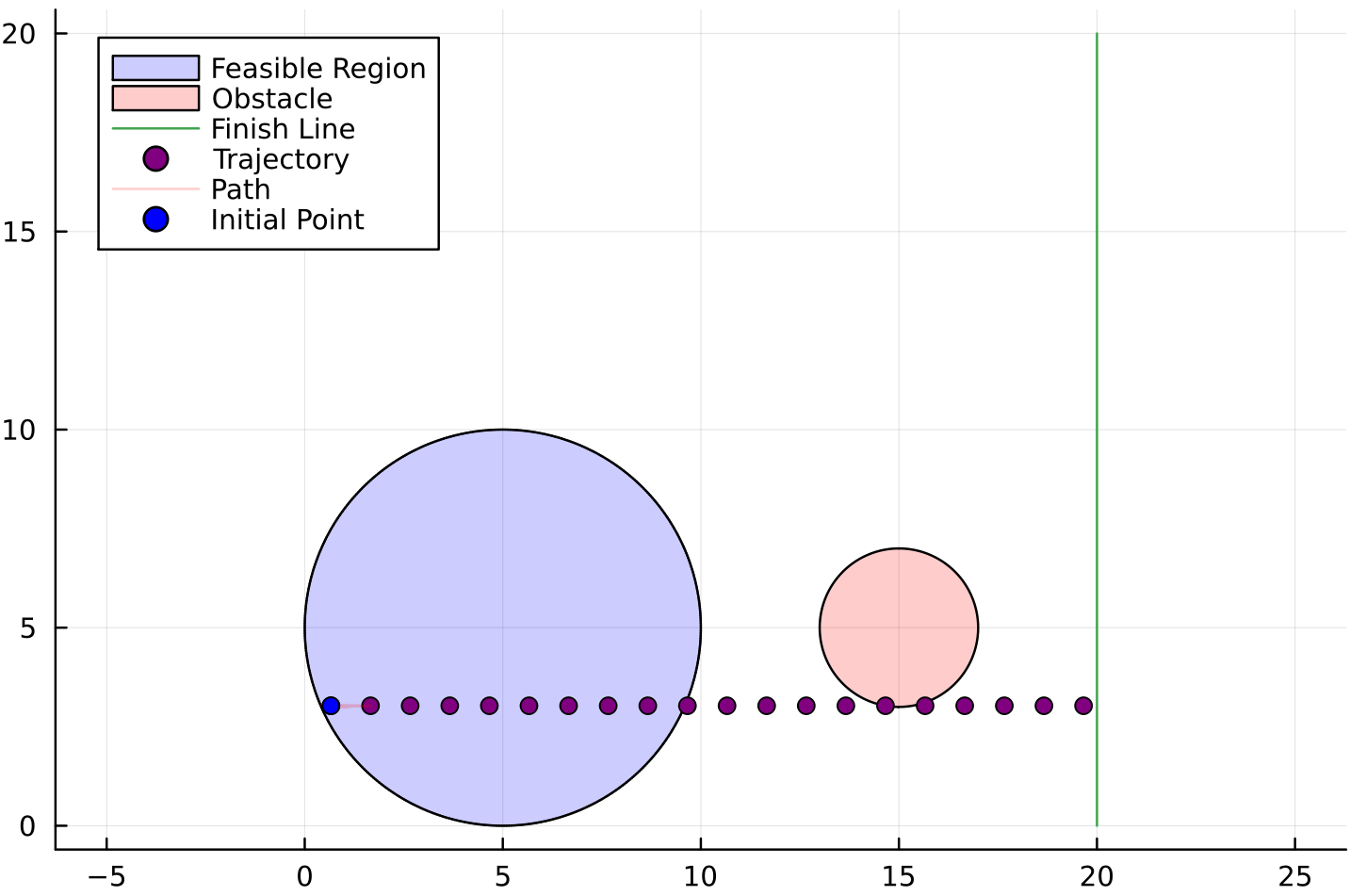}\par
\includegraphics[width=0.75\linewidth]{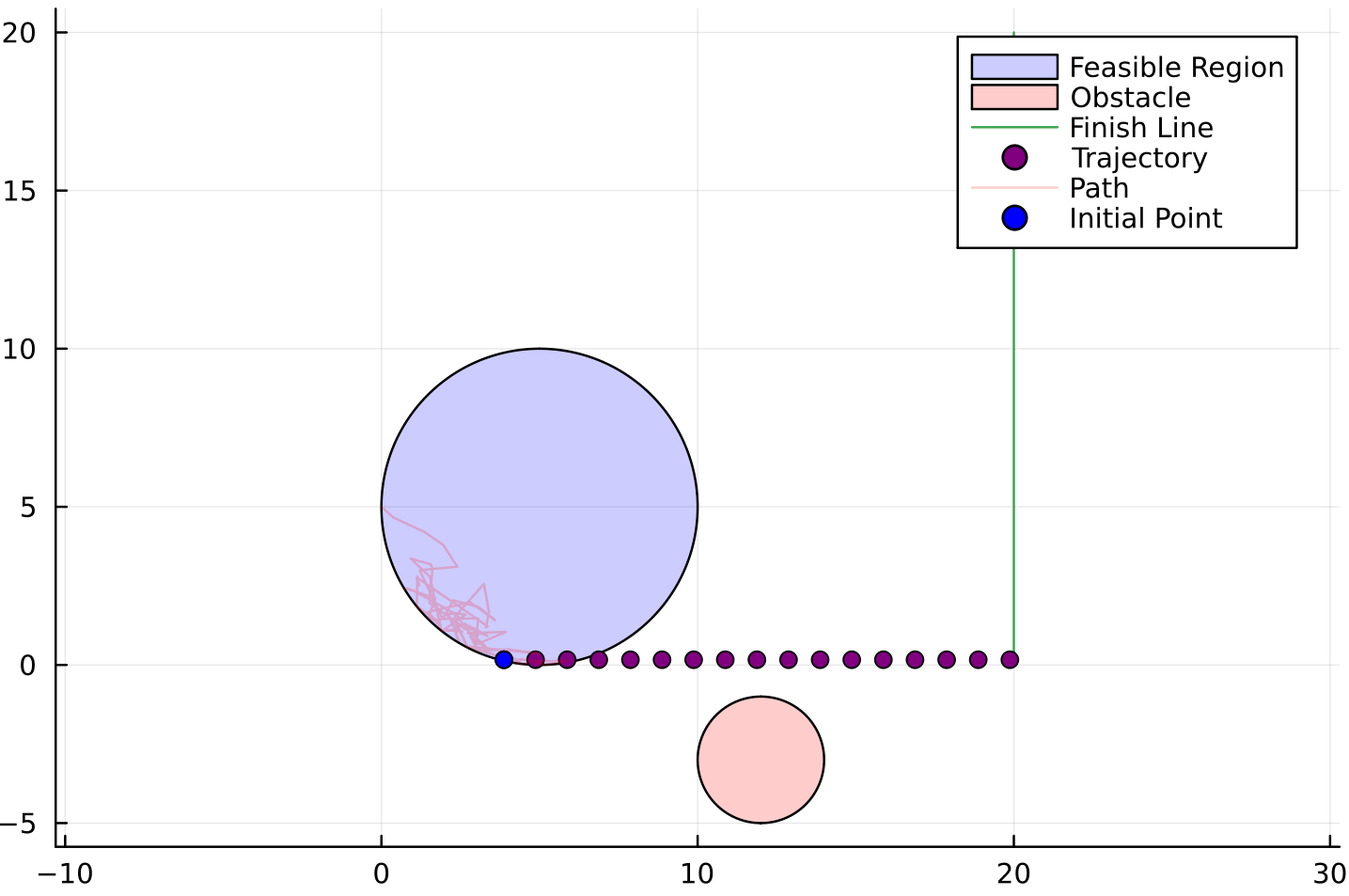}\par
\end{multicols}
\caption{Adversarial Initial points for linear trajectories for obstacle center $o=(15, 13)$ (top-left), $o=(12, 9)$ (top-right), $o=(15, 5)$ (bottom-left), and $o=(12, -3)$ (bottom-right). Time taken for the solutions is, in order, 220.8 seconds, 269.47 seconds, 262.9 seconds, and 267.39 seconds. Differences in timing indicate that the final level solver converged quickly for some instances of the problem. Ideal points are as left as possible and are either touching the obstacle or come as close to touching it as possible. Red lines indicate the path of solution in the feasible region.}
\label{fig:linear}
\end{figure*}

\begin{figure*}
\centering
\begin{multicols}{2}
\includegraphics[width=0.75\linewidth]{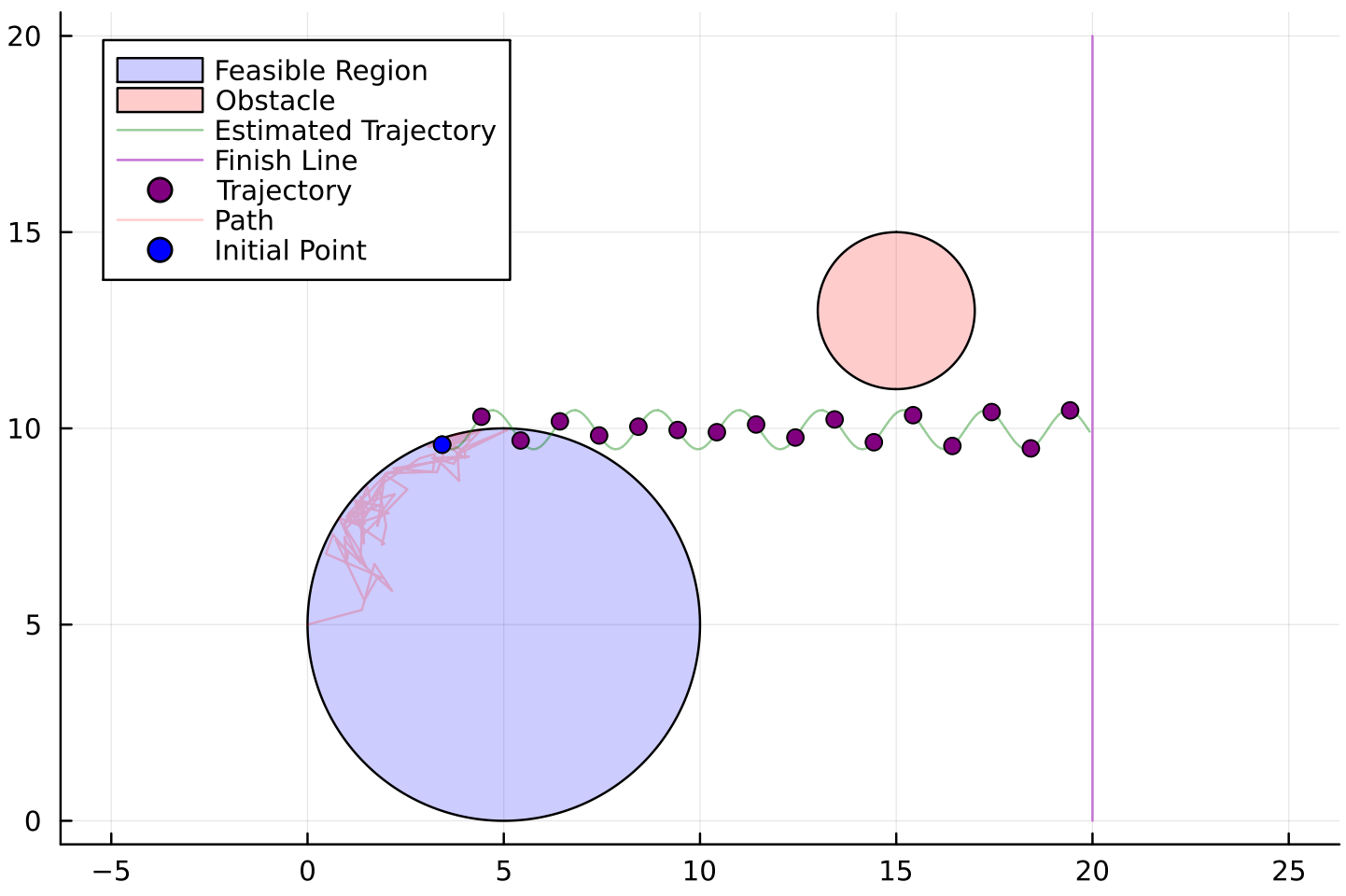}\par
\includegraphics[width=0.75\linewidth]{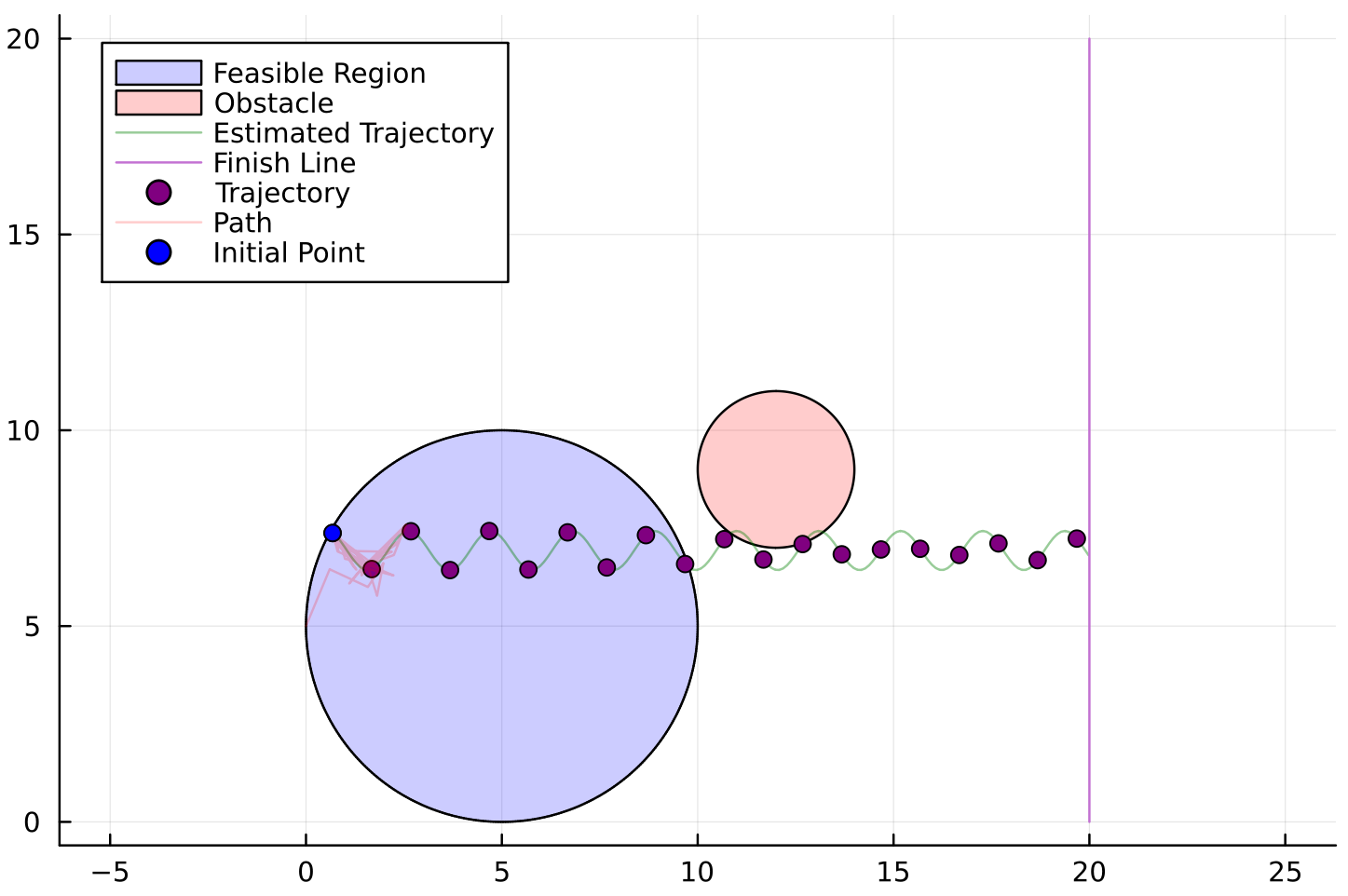}\par
\end{multicols}
\begin{multicols}{2}
\includegraphics[width=0.75\linewidth]{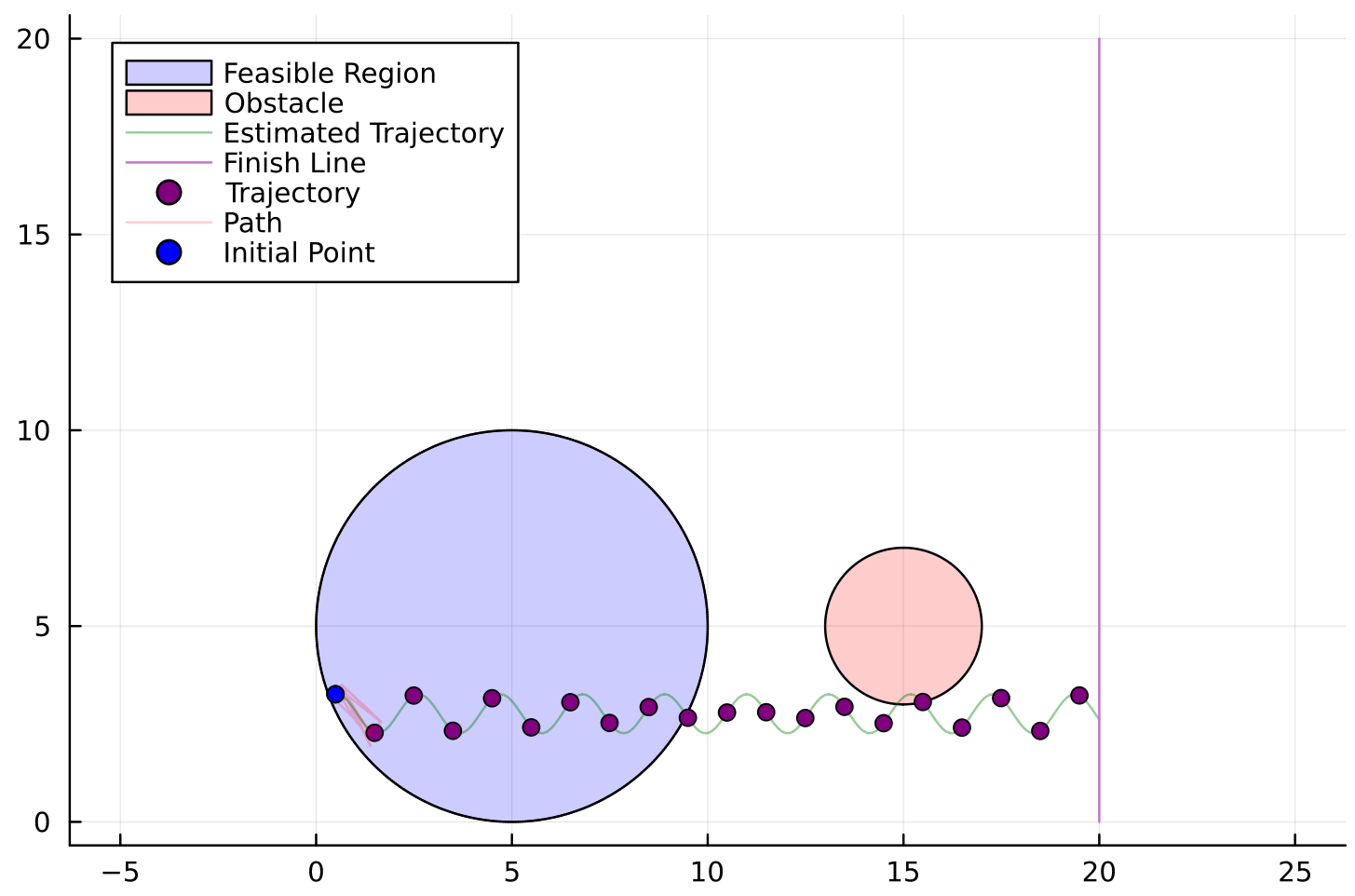}\par
\includegraphics[width=0.75\linewidth]{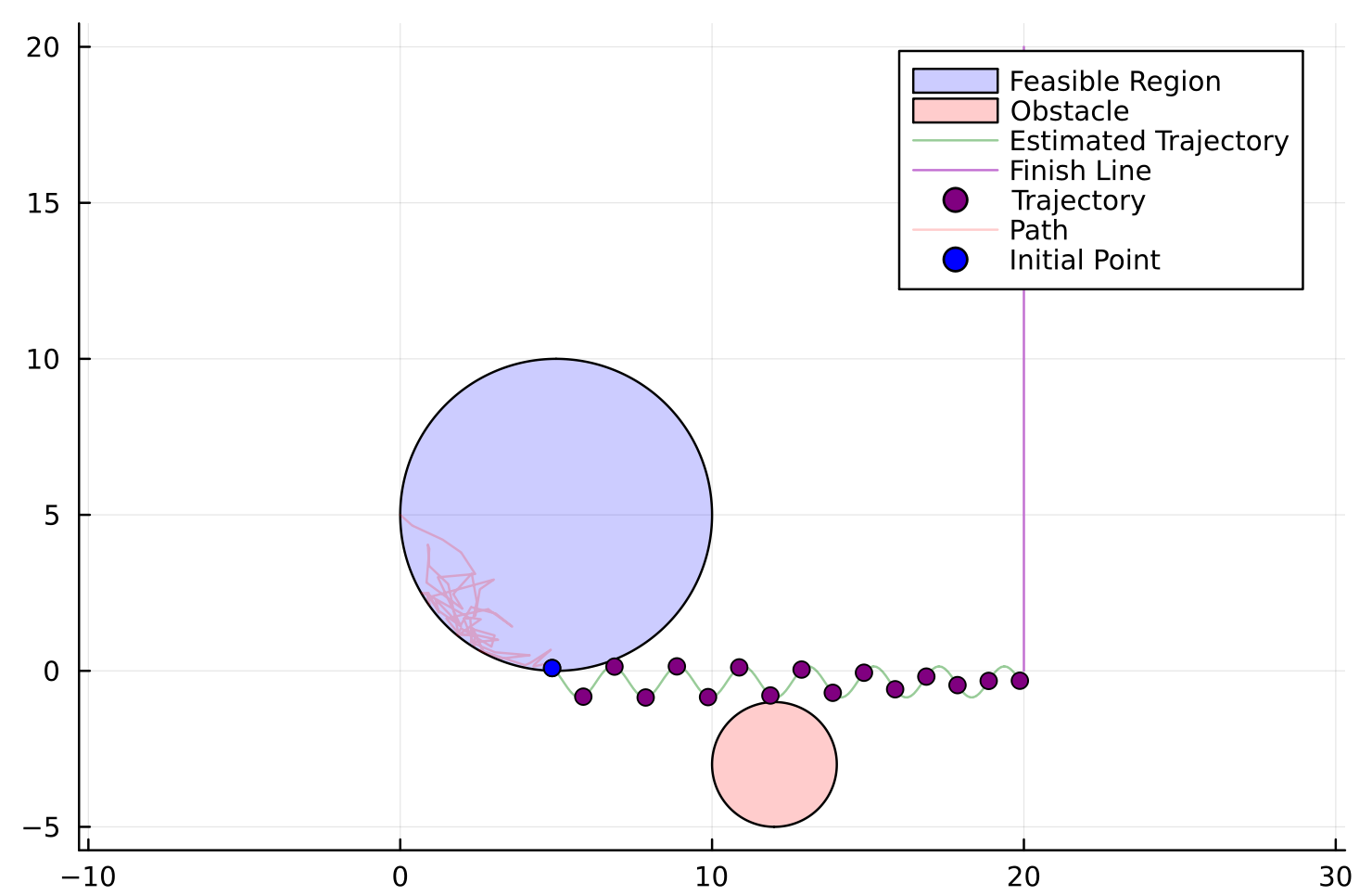}\par
\end{multicols}
\caption{Adversarial Initial points for nonlinear (sinusoidal) trajectories for the obstacle center are $o=(15, 13)$ (top-left), $o=(12, 9)$ (top-right), $o=(15, 5)$ (bottom-left), and $o=(12, -3)$ (bottom-right). The time taken for the solutions, in order, is 2272.13 seconds, 2624.19 seconds, 2408.66 seconds, and 2356.23 seconds. Ideal points are as left as possible and are either touching the obstacle or come as close to touching it as possible. Red lines indicate the path of the solution in the feasible region. All the solutions obtained are quite close to optimality. The top-left instance may not look optimal but the phase of the sinusoid and our sampling strategy may go counter to our intuition.}
\label{fig:nonlinear}
\end{figure*}

\subsection{Discussion on Results}
As can be seen, in all cases, MCMO generates proper adversarial initial conditions for this problem. For linear policy (figure \ref{fig:linear}), except the bottom-right setup, all other instances of the problem achieve optimal results. For the bottom-right instance, although the obtained point is not optimal, the error is $\approx 7.5\%$ which is not at all unreasonable considering the stochatic nature of the algorithm. However, accuracy can be further increased to desired bounds by running MCMO with a higher number of samples and sampling iterations and lower values for $\alpha$ for further iterations. The path taken by the solution at each iteration is traced by the red line. Unsurprisingly, for problems where the initial feasible solutions were closer to optimality, the algorithm converged to the answer in very few iterations. For problems where the initially feasible solution was not close to optimality, the path appears repetitive and chaotic, eventually converging to the answer, but it must be taken into account that the plotted path is a projection $(x_1, x_2)$ of the true decision space $(x_1, x_2, T)$.

For non-linear policy (figure \ref{fig:nonlinear}), almost all of the instances converge to the optimum. While the true trajectory (represented by the sinusoid) does intersect the obstacle, it is to be expected because our problem formulation has been for the discrete samples to begin with, which incidentally behave as expected. Furthermore, it may also appear that the top-right instance of the problem is not optimal, as the generated trajectory is not as close to the obstacle as possible. But owing to the fact that the sinusoid's phase depends upon the initial point and taking into account our sampling strategy, moving the point to the top does not, in fact, bring the trajectory any closer to the obstacle.

\subsection{Solving Nested Toll-Setting Problem using MCMO}
\label{subsec:nestedresults}
A comparison of the known solutions for the edge cases with the solutions obtained by MCMO is tabulated in Table \ref{table:nested}. The parameters used for both of the instances of the problem are $N^1=7, N^2=7, X_s=[0, 0, 1, 0, 0], \alpha=0.15, maxiter=100$. Smoothing scheme used is best objective smoothing with $k=10$. From the table, we can see that MCMO achieves results with error (w.r.t $f_1$) of $0.025\%~ \text{and} ~1.76\%$ respectively for parameters $D=6,-1.5$ taking, respectively, 119.11 and 129.73 seconds. The achieved results are quite satisfactory but can be made more accurate by decreasing the step sizes $\alpha$ and increasing the number of samples and sampling iterations as required.

\newcolumntype{~}{>{\global\let\currentrowstyle\relax}}
\newcolumntype{^}{>{\currentrowstyle}}
\newcommand{\rowstyle}[1]{\gdef\currentrowstyle{#1}%
#1\ignorespaces
}
\begin{table}[]
\centering
\small
\setlength{\tabcolsep}{4pt}
\begin{NiceTabular}{|l|l|l|l|l|l|l|l|}
\hline
\RowStyle[bold]{}
$D$ & $X^*$ & $t_1$ & $t_2$ & $p_1$ & $p_2$ & $p_3$ & $f_1$ \\ \hline
$6$ & $\olsi X$ & $\ge$ 2 & 4 & 0 & 1 & 0 & 4 \\
& $X_{M}$ & 2.387 & 4.032 & 0.006& 0.989 & 0.005 & 4.001 \\ \hline
$-1.5$ & $\underline X$ & 1 & $\ge$0 & 0.25 & 0 & 0.75 & 0.25 \\
& $X_{M}$ & 1.281 & 0.15 & 0.199 & 0 & 0.801 & 0.254 \\ \hline
\end{NiceTabular}
\caption{A comparison of solution obtained via MCMO $X_M$ with known analytical solution $X^*$ for the Nested Toll-Setting Problem for different parameters $D=\olsi D = 6, \text{and} ~D = \underline D = -1.5$. The achieved errors are $0.025\%$ and 1.76\% respectively.}
\label{table:nested}
\end{table}
\subsection{Numerical Examples from the Literature}
The following problem derived from \cite{sinha2003fuzzy} is a trilevel linear problem defined as:
\begin{subequations}
\begin{align*}
\label{eq:sinha}
&\max_{x_1, x_2,x_3,x_4}~~7x_1 + 3x_2 -4x_3+2x_4 \\
&~~~~ s.t.~x_3,x_4\in\arg\max_{x_3, x_4}~ x_2 + 3x_3 + 4x_4\\
&~~~~~~~~~~~s.t.~x_4\in\arg\max_{x_4} ~ 2x_1+x_2+x_3+x_4\\
&~~~~~~~~~~~~~s.t.~x_1+x_2+x_3+x_4\le 5 \\
&~~~~~~~~~~~~~~~~~~x_1+x_2-x_3-x_4\le2 \\
&~~~~~~~~~~~~~~~~~~x_1+x_2+x_3 \ge 1\\
&~~~~~~~~~~~~~~~~~~-x_1+x_2+x_3\le1\\
&~~~~~~~~~~~~~~~~~~x_1-x_2+x_3+2x_4\le 4\\
&~~~~~~~~~~~~~~~~~~x_1+2x_3+3x_4\le3\\
&~~~~~~~~~~~~~~~~~~x_4\le2\\
&~~~~~~~~~~~~~~~~~~x_1,x_2,x_3,x_4 \ge 0
\end{align*}
\end{subequations}

The optimum $f_1^* = 16.25$ for this problem is reported at $(2.25, 0, 0, 0.25)$. MCMO obtains a result of $16.145$ at $(2.205, 0.06, 0, 0.265)$ when run with the following parameters: $N^1=6, N^2=3,M^1=M^2=1, X_s=[0.4, 0.4, 0.4, 0.4], \alpha=1, maxiter=100$. The sample size was chosen owing to the difference in the number of variables, while the feasible set was deduced by observation. The smoothing scheme used is the best objective smoothing with $k=10$. The relative error in objective values for this example is $< 1\%$. The time taken to obtain the solution is 59.26 seconds.

The second problem is taken from \cite{tilahun2012new} and is defined as:
\begin{align*}
&\min_{x,y,z}~~-x + 4y \\
&~~~~ s.t. ~x+y \le 1\\
&~~~~ y,z\in\arg\min_{y,z}~ 2y+z\\
&~~~~~~~~~~s.t.~ -2x+y \le -z\\
&~~~~~~~~~~~z\in\arg\min_{z} ~ -z^2+y\\
&~~~~~~~~~~~~~~~~~~s.t.~~z\le x; x\in[0, 0.5]; y\in[0,1]; \\
&~~~~~~~~~~~~~~~~~~~~~~~~~~z\in[0,1]\\
\end{align*}
The reported optimum $f_1^*=-0.5$ is at $(0.5, 0, 0.0095)$ whereas MCMO obtains $f_1=-0.498$ at $(0.498, 0, 0.498)$ when run with the following parameters: $N^1=5, N^2=5, M^1=M^2=1, X_s=[0, 0, 0], \alpha=0.2, maxiter=100$. The choice of $\alpha$ was guided by the bounds on the decision variables, and the initial feasible solution was obtained via observation. The smoothing scheme used is the best objective smoothing with $k=10$.
The obtained minimizer disagrees with the reported minimizer, but it can be seen that the reported minimizer, i.e., $(0.5, 0, 0.0095)$ is incorrect as opposed to the actual minimizer, i.e., $(0.5, 0, 0.5)$ because once $x,y=0.5, 0$ are chosen, $z$ can be clearly increased (upto $x$) by the last player to achieve further minimality. Moreover, \cite{woldemariam2015systematic} agrees with our results on the same problem and reports $f_1=-0.4929$ at $(0.4994, 0.0016, 0.4988)$. For this problem, MCMO achieves a relative error of $< 1\%$ in 20.85 seconds.

\section{Comparisons}
\label{sec:comparisons}
\subsection{Comparisons with Existing Works}
We compare the results obtained in subsection \ref{subsec:nestedresults} for the Nested Toll-Setting problem with some of the existing methods from the literature. We chose \cite{tilahun2012new} and \cite{woldemariam2015systematic} as baselines because these methods have been proposed for arbitrarily deep multilevel optimization problems as well. Since none of these methods are capable of solving problems with equality constraints, we have to reformulate the Nested Toll-Setting problem to remove the equality constraint as follows:
\begin{align*}
&\max_{t_1, t_2,p_1,p_2}~~p_1 \cdot t_1 + p_2 \cdot t_2 \\
&~~~~ s.t.~ t_1, t_2 \in [0,10] \\
&~~~~ p_1,p_2 \in\arg\min_{p_1,p_2}~~p_1 \cdot (p_1 + t_1) + (1-p_1)^2\\
&~~~~~~~~~~~~~~~~~~ s.t. ~p_1 \in [0, 1] \\
&~~~~~~~~~~~~~~~~~~~~~~~~p_2\in\arg\min_{p_2} ~ p_2 \cdot (p_2+t_2) ~+ \\
&~~~~~~~~~~~~~~~~~~~~~~~~~~~~~~~~~(1-p_1-p_2) \cdot (1-p_1-p_2 + D)\\
&~~~~~~~~~~~~~~~~~~~~~~~~~~~~~~~~~~s.t.~p_2 \in [0,1-p_1]\\
\end{align*}
We note that, in general, it may not always be possible to reformulate a problem to remove a particularly tricky equality constraint. For both methods, we generate an extremely high number of samples per iteration, i.e., $10^6$ for the entire decision space, and run both algorithms for 100 iterations. For $D=-1.5$ \cite{tilahun2012new}'s method obtains $f_1=2.53\times 10^{-7}$ at $X=(27.65, 36.93, 1.21\times10^{-9}, 5.95\times10^{-9})$ and for $D=6$, it obtains $f^1=2.43\times10^{-6}$ while taking 875.2 seconds and 883.41 seconds respectively. These results have a very high error compared to the theoretical best, and the reason for this is that this algorithm does not truly solve a Multilevel Stackelberg problem at all. It's instead an Iterative Best Response type algorithm, which only works when finding the Nash Equilibrium of a problem.

\cite{woldemariam2015systematic}'s approach only works for bounded decision variables, so we add two additional constraints, i.e., $t_1\in[0,10], t_2\in[0,10]$. For $D=6$, it obtains $f_1=4.75$ at $X=(9.64, 5.13, 0.066, 0.8)$ in 500.76 seconds, which overestimates the theoretical maximum by the relative error of $18.75\%$ and for $D=-1.5$, it obtains $f_1=0.729$ at $X=(9.8, 6.72, 0.066, 0.0113)$ which has a relative error of $191.6\%$ in about 465.02 seconds. Even though this method works much better than the former, it still ends up overestimating the leader's objective most of the time. This is because of the update rule of this algorithm, which only ever changes the obtained solution if it's better than the previous one for just the leader. So if a solution with high complementary error but better leader objective is acquired, it's always kept regardless of whatever may be found in subsequent iterations.
\newcommand{\column}[1]{%
\begin{bmatrix}
#1
\end{bmatrix}
}
\newcommand{\norm}[1]{\left\lVert~#1~\right\rVert_2}

\subsection{Timing Comparison for arbitary levels}
To compare the time required by MCMO to solve any given problem against its complexity, we introduce the following arbitrarily multilevel problem parameterized by $w\in(\real^+)^{n}$.

\begin{align*}
&\min_{x_1,...,x_n\in\real^n}~~ \norm{\column{x_1 \\ x_2 \\ \vdots \\x_n}-\column{w_1\\ w_2\\ \vdots\\ w_n}}\\
&~~~~~~~~~~~~ x_2, ...x_n \in \arg\min_{x_2,...,x_n\in\real^{n-1}}~~\norm {\column{x_2 \\ \vdots \\x_n} - \column{w_2\\ \vdots\\ w_n}}\\
&~~~~~~~~~~~~~~~~~~~~~~~~~~~~~~~~ s.t.~~x_2 \le x_1\\
&~~~~~~~~~~~~~~~~~~~~~~~~~~~~~~~~~~~~~~~ \vdots\\
&~~~~~~~~~~~~~~~~~~~~~~~~~~~~~~~~~~~~~~x_n \in\arg\min_{x_n\in\real} ~ (x_n-w_n)^2\\
&~~~~~~~~~~~~~~~~~~~~~~~~~~~~~~~~~~~~~~~~s.t. ~~ x_n\le x_{n-1}
\end{align*}

While this problem generally degenerates to a single-level problem of the form $min_x \norm{x-w}; x_i\le x_{i-1} \forall i$, it's still an ideal problem to test our algorithm for the fact that the dimension of the decision variable $x$ increases linearly with the level. When solved with the parameters $M^i=1, \alpha_i=0.25 \forall i$, for 50 iterations each for different values of sample sizes per level $N^i$, the obtained timing results are tabulated in Table \ref{tab:my_label} and the corresponding graph is shown in Figure 5. The execution time grows exponentially as the levels increase, which is as expected of a recursive algorithm.

\begin{table}
\centering
\begin{tabular}{|c|c|c|c|c|} \hline
& \multicolumn{4}{|c|}{\textbf{Time (s)}}\\ \hline
\textbf{Levels}& \textbf{N=3}&\textbf{N=4}&\textbf{N=5}&\textbf{N=6}\\ \hline
2& 1.33& 1.68& 1.71&2.33\\ \hline
3& 3.84& 5.52& 7.71&11.86\\ \hline
4& 10.03& 21.57& 39.41&77.194\\ \hline
5& 29.84& 101.81& 433.91&1356.82\\ \hline
6&97.15& 1230.33& 7,463.52&10938.59\\ \hline
\end{tabular}
\caption{Time taken as Levels increase for different sample sizes $N$}
\label{tab:my_label}
\end{table}

\begin{figure}
\label{fig:timing}
\begin{tikzpicture}[scale=0.85]
\begin{axis}[
title={Problem Complexity vs Time for different $N^i$},
xlabel={Number of Levels},
ylabel={Time for 50 Iterations [s]},
xmin=0, xmax=7,
ymin=0, ymax=5000,
xtick={1,2,3,4,5,6,7},
ytick={0,1000, 2000, 3000, 4000,5000,6000,7000,8000,9000},
legend pos=north west,
ymajorgrids=true,
grid style=dashed,
]
\addplot[
color=red,
mark=square,
]
coordinates {
(2, 1.33)(3, 3.84)(4, 10.03)(5, 29.84)(6, 97.15)
};
\addplot[
color=green,
mark=square,
]
coordinates {
(2, 1.68)(3, 5.52)(4, 21.57)(5, 101.81)(6, 1230.33)
};
\addplot[
color=blue,
mark=square,
]
coordinates {
(2, 1.71)(3, 7.71)(4, 39.41)(5, 433.91)(6, 7463.52)
};
\addplot[
color=purple,
mark=square,
]
coordinates {
(2, 2.33)(3, 11.86)(4, 77.194)(5, 1356.82)(6, 10938.59)
};
\legend{$N^i=3$, $N^i=4$, $N^i=5$, $N^i=6$}

\end{axis}
\end{tikzpicture}
\caption{Execution time as problem complexity increases for different sample sizes. The obtained graph is exponential as the problem increases linearly, which is as expected.}
\end{figure}
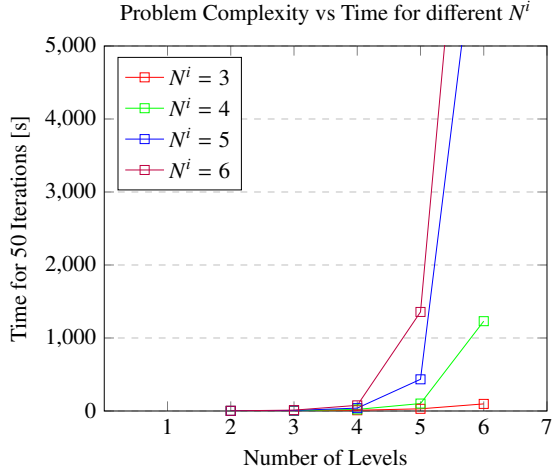

\subsection{Accuracy Comparison}
In general, we expect the accuracy of MCMO to increase as the number of samples per level $N^i$ increases. For this experiment, we use a five level version of the problem introduced in the previous subsection with a randomly generated $w = (3, 8, 7, 7, 3)$. We fix $M^i=1, \alpha^i=0.25$ and start with the initial guess $x_s=(0,0,0,0,0)$ and plot the convergence of the algorithm per iteration in Figure 6. As can be noticed, as the sample sizes increase, the convergence of the algorithm does increase, but it gets capped beyond a certain point because of the step size $\alpha^i$. The effect of increasing the step size for the same problem by keeping the sample sizes fixed at $N^i=6$ for different $\alpha^i=0.25, 0.5, 1$ is shown in Figure 7. It can be observed that increasing $\alpha^i$ for an appropriate number of samples increases the convergence speed to the optimum. Once the optimum is approached the convergence plateaus. This demonstrates that MCMO is stable, i.e., once it approaches the neighborhood of a stable solution, it remains there (provided enough samples are taken at each level).

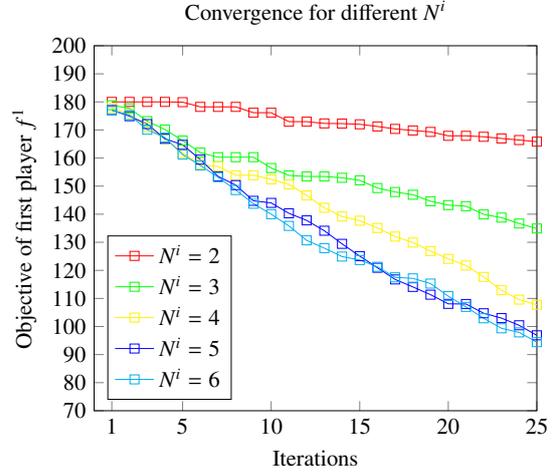
\begin{figure}
\label{fig:accuracy}
\begin{tikzpicture}[scale=0.85]
\begin{axis}[
title={Convergence for different $N^i$},
xlabel={Iterations},
ylabel={Objective of first player $f^1$},
xmin=0, xmax=25,
ymin=70, ymax=200,
xtick={1,5,10,15,20,25},
ytick={0,10,20,30,40,50,60,70,80,90,100,110,120,130,140,150,160,170,180,190,200},
legend pos=south west,
]
\addplot[
color=red,
mark=square,
]
coordinates {
(1,179.999999942506)
(2,179.999999942506)
(3,179.999999942506)
(4,179.999999942506)
(5,179.906098007488)
(6,178.229299423607)
(7,178.229299423607)
(8,178.229299423607)
(9,176.147235002161)
(10,176.147235002161)
(11,172.958768413922)
(12,172.958768413922)
(13,172.333700077797)
(14,172.275815319771)
(15,172.020537190578)
(16,171.252270492493)
(17,170.378798320103)
(18,169.827023329116)
(19,169.316554381522)
(20,167.920177233806)
(21,167.920177233806)
(22,167.575821792866)
(23,166.9786843419)
(24,166.412534406513)
(25,165.837043765998)

};

\addplot[
color=green,
mark=square,
]
coordinates {
(1,178.820312035935)
(2,177.697758940571)
(3,173.155953233089)
(4,170.162486356941)
(5,166.232746745412)
(6,161.92645659303)
(7,160.304002687388)
(8,160.304002687388)
(9,160.304002687388)
(10,156.449589620247)
(11,153.80867092927)
(12,153.432685119472)
(13,153.432685119472)
(14,152.967829953642)
(15,152.066118611413)
(16,149.311765225477)
(17,147.905453069716)
(18,146.936409954907)
(19,144.607601686924)
(20,143.240786640583)
(21,142.861434272013)
(22,139.907454215163)
(23,138.791586856514)
(24,136.677692187752)
(25,134.899933717834)

};

\addplot[
color=yellow,
mark=square,
]coordinates{
(1,177.834662336204)
(2,175.599575242034)
(3,171.053240355354)
(4,166.728682948039)
(5,161.816587875231)
(6,157.630997172662)
(7,157.023494559204)
(8,153.875693890327)
(9,153.875693890327)
(10,152.395895234977)
(11,150.620211247164)
(12,146.67914033351)
(13,142.276934020807)
(14,139.188772728556)
(15,137.679176601163)
(16,135.089119013311)
(17,132.028314033516)
(18,129.913461632603)
(19,126.861505614253)
(20,124.077643367852)
(21,121.738989976475)
(22,117.653816952603)
(23,112.923583824496)
(24,109.583063644128)
(25,107.733869656874)
};

\addplot[
color=blue,
mark=square,
]coordinates{
(1,177.096904471805)
(2,174.904017793307)
(3,172.100264490269)
(4,166.86214632643)
(5,164.710676101897)
(6,159.446785158958)
(7,153.486531534488)
(8,150.356172533257)
(9,144.758325817505)
(10,144.0271838961)
(11,140.315278944645)
(12,137.801182649402)
(13,134.113263048204)
(14,129.443235465722)
(15,125.067137439355)
(16,120.926645151734)
(17,116.762800744671)
(18,114.101858446831)
(19,111.332500490264)
(20,108.069546669344)
(21,108.017441870698)
(22,104.721305693133)
(23,102.99549492505)
(24,100.443248848969)
(25,96.8674537089256)

};

\addplot[
color=cyan,
mark=square,
]coordinates{(1,176.989432074654)
(2,175.272012012475)
(3,170.078685527855)
(4,167.322360262341)
(5,161.22668018116)
(6,157.359136937528)
(7,153.182920989661)
(8,148.60472189542)
(9,143.726451917734)
(10,140.008306004607)
(11,135.824976116912)
(12,130.685746217479)
(13,127.913308403177)
(14,124.911910574068)
(15,123.663831743742)
(16,121.262586233099)
(17,117.505090719225)
(18,117.148660227965)
(19,115.288385974704)
(20,110.81709941656)
(21,106.978792228646)
(22,102.928171687576)
(23,99.3656532609362)
(24,97.903794094871)
(25,94.4344284307565)
};
\legend{$N^i=2$, $N^i=3$, $N^i=4$, $N^i=5$, $N^i=6$}

\end{axis}
\end{tikzpicture}
\caption{Convergence for different sample sizes. After a certain threshold, step size $\alpha^i$ caps the performance.}
\end{figure}

\begin{figure}
\label{fig:accuracy2}
\begin{tikzpicture}[scale=0.85]
\begin{axis}[
title={Convergence for different $\alpha^i$},
xlabel={Iterations},
ylabel={Objective of first player $f^1$},
xmin=0, xmax=25,
ymin=0, ymax=200,
xtick={1,5,10,15,20,25},
ytick={0,10,20,30,40,50,60,70,80,90,100,110,120,130,140,150,160,170,180,190,200},
legend pos=north east,
]
\addplot[
color=red,
mark=square,
]
coordinates {
(1,176.989432074654)
(2,175.272012012475)
(3,170.078685527855)
(4,167.322360262341)
(5,161.22668018116)
(6,157.359136937528)
(7,153.182920989661)
(8,148.60472189542)
(9,143.726451917734)
(10,140.008306004607)
(11,135.824976116912)
(12,130.685746217479)
(13,127.913308403177)
(14,124.911910574068)
(15,123.663831743742)
(16,121.262586233099)
(17,117.505090719225)
(18,117.148660227965)
(19,115.288385974704)
(20,110.81709941656)
(21,106.978792228646)
(22,102.928171687576)
(23,99.3656532609362)
(24,97.903794094871)
(25,94.4344284307565)

};

\addplot[
color=green,
mark=square,
]
coordinates {
(1,174.011058686556)
(2,170.632554787861)
(3,160.491725468431)
(4,155.243285278679)
(5,143.716606951587)
(6,136.593390644281)
(7,128.879530116202)
(8,120.722282046825)
(9,112.264421114454)
(10,105.952020634007)
(11,99.0056636269257)
(12,90.7183241806345)
(13,86.3573723105384)
(14,81.8818908051875)
(15,79.9855635766004)
(16,76.3046152087433)
(17,70.8040111554374)
(18,70.3239499678501)
(19,67.6384834269641)
(20,61.6131020077344)
(21,56.7835588815492)
(22,51.9834242851695)
(23,47.658681971112)
(24,45.74322069603)
(25,41.813079468447)

};

\addplot[
color=blue,
mark=square,
]coordinates{
(1,168.150895349617)
(2,161.619232454885)
(3,142.320868415262)
(4,132.88082940086)
(5,112.486200087759)
(6,100.687248192978)
(7,87.8138126074409)
(8,75.4959169451073)
(9,63.7749111724557)
(10,55.6456755954625)
(11,47.3249286916404)
(12,38.1752674350704)
(13,34.1890587289117)
(14,30.4625122846397)
(15,28.4878532700744)
(16,26.8204500697092)
(17,22.5100327005517)
(18,22.09226053819)
(19,20.4683880907583)
(20,17.5868267861063)
(21,16.6951206247077)
(22,16.4793117826534)
(23,15.1219339880448)
(24,15.0846338500161)
(25,15.0846338500161)

};

\addplot[color=cyan]coordinates{
(1,14.75)
(2,14.75)
(3,14.75)
(4,14.75)
(5,14.75)
(6,14.75)
(7,14.75)
(8,14.75)
(9,14.75)
(10,14.75)
(11,14.75)
(12,14.75)
(13,14.75)
(14,14.75)
(15,14.75)
(16,14.75)
(17,14.75)
(18,14.75)
(19,14.75)
(20,14.75)
(21,14.75)
(22,14.75)
(23,14.75)
(24,14.75)
(25,14.75)

};

\legend{$\alpha^i=0.25$, $\alpha^i=0.5$, $\alpha^i=1$, True}

\end{axis}
\end{tikzpicture}
\caption{Convergence for different step sizes for $w=(3,8,7,7,3)$. It can be shown that the minimum for this problem is at $X=(6.25, 6.25, 6.25, 6.25, 3)$ with leader's objective $f^1=14.75$ (shown in the plot by the cyan horizontal line)}
\end{figure}
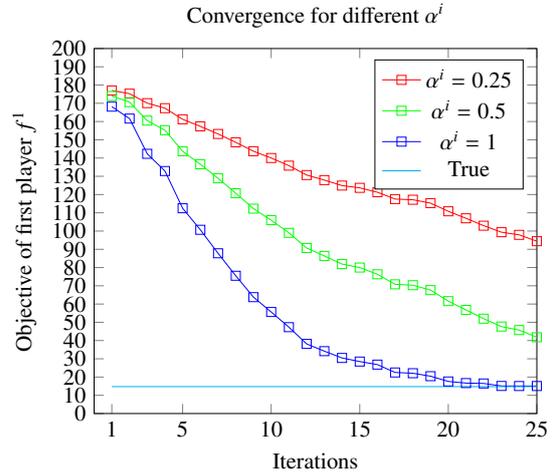

\section{Conclusion and Future Work}
\label{section:conclusions}
Stackelberg games arise in many real-world scenarios, and conversely, many interesting economic, control, and other causal phenomena can be naturally modeled as Stackelberg games. Multilevel Stackelberg games provide a further generalization that expands the perimeter of interesting interactions that can be modeled by such rules. However, the difficulty involved in solving them is non-trivial and can present a major challenge to those who seek to model and solve these kinds of problems. In this paper, we introduced two such example problems that can effortlessly be modelled using multilevel formulation, i.e., a) the Adversarial Initial Condition determination problem, where we find a challenging initial condition for any provided policy, and b) the Nested Toll-Setting problem, which is a generalization of the famous Bilevel Toll-Setting problem. We then presented MCMO, a stochastic algorithm that can be used to solve problems involving an arbitrary number of leaders and followers (i.e., arbitrarily deep multilevel games) up to desired accuracy, and presented proofs for its feasibility and (under certain assumptions) optimality. We then used this algorithm to solve the multilevel problems we constructed and also solved a few problems from the literature for comparison, achieving satisfactory results in each case.

Future work in this direction would be to improve the convergence speed and accuracy of this algorithm. Furthermore, a desired generalization of this algorithm would be one that works with multiple leaders and multiple followers at all levels (or the so-called Multilevel Decentralized Problem). This would enable us to solve a wider variety of interesting problems that involve numerous stakeholders with varying levels of power amongst themselves. And finally, for applications where an exact solution is required, we want to explore methods to obtain them by leveraging the approximate solution provided by MCMO.

\bibliographystyle{elsarticle-harv}
\bibliography{cas-refs}

\appendix
\label{appendixA}
\section{Nested Toll-Setting Problem}
The third level can be reformulated as:
\begin{subequations}
\begin{align}
\min_{p_2} ~&p_2 \cdot (p_2 + t_2) + \\&(1-p_1-p_2) \cdot (1-p_1-p_2+D)\\
&0\le p_2 \le 1-p_1
\end{align}
\end{subequations}

The unconstrained stationary point for this problem is when:
\begin{subequations}
\begin{align}
\nonumber p_2 + p_2 + &t_2 - (1-p_1-p_2)\\&- (1-p_1-p_2+D) = 0\\
&4p_2 + 2p_1 +t_2-D-2=0\\
&p_2 = \frac{D+2-2p_1-t_2}{4}
\end{align}
\end{subequations}

From the theory of constrained minimization, the response of the third level can then be written as:

\begin{align}
\label{eq:p2response}
p_2(t_1, t_2) = \begin{cases}
0 & \text{if} ~ D+2-2p_1(t_1)\le t_2 \\
1-p_1(t_1) & \text{if} ~D-2+2p_1(t_1)\ge t_2 \\
\frac{2+D-2p_1(t_1)-t_2}{4} & \text{otherwise}
\end{cases}
\end{align}

Similarly, we can obtain the response of the second level as:
\begin{align}
\label{eq:p1response}
p_1(t_1) = \begin{cases}
0 & \text{if} ~ t_1 \ge 2 \\
1 & \text{if} ~ t_1 \le -2\\
\frac{1}{4}(2-t_1) & \text{otherwise}
\end{cases}
\end{align}

From equations \ref{eq:p1response} and \ref{eq:p2response}, we can define the following parameterized constraint sets:
\begin{align*}
\mathcal{C}_1(D) := \{ & p_1=0 \wedge t_1 \ge 2, \\
& p_1=1 \wedge t_1\le-2, \\
&p_1=\frac{1}{4}(2-t_1) \wedge -2< t_1< 2\}
\end{align*}
\begin{align*}
\mathcal{C}_2(D) := \{ &p_2=0 \wedge ~ D+2-2p_1\le t_2 \\
&p_2=1-p_1 \wedge ~D-2+2p_1(t_1)\ge t_2 \\
&p_2=\frac{2+D-2p_1(t_1)-t_2}{4} \wedge \\
&~~~~~~~D-2+2p_1(t_1)< t_2 < D+2-2p_1 \}
\end{align*}

The solution for the Nested toll-setting problem would then simply be:
\begin{align}
\label{finaleq}
\max_{t_1, t_2, p_1, p_2} &p_1 \cdot t_1 + p_2 \cdot t_2 \\
\nonumber &t_1 \ge 0, t_2 \ge 0\\
\nonumber &(t_1,t_2,p_1,p_2) \in \bigcup \mathcal{C}_1(D) \cdot \mathcal{C}_2(D)
\end{align}

Equation \ref{finaleq} is a standard quadratic programming problem defined over a union of polyhedral regions. It can be solved for each of the polyhedral regions using a standard solver to obtain the optimum value for the problem as follows:
\begin{itemize}
\item For $D = 6$, the obtained maximum is $4$ for $p_1=0, t_1=14, p_2=1, t_2=4$. The results imply that the toll-setter benefits when no traffic takes the tolled road at station $T_1$, i.e., ($p_1=0$) and all traffic takes the tolled road at station $T_2$.
It can be seen from equation \ref{eq:p1response}a that the same objective can be realized for any $t_1\ge2$ (as this makes $p_1=0$).

\item For $D = -1.5$, obtained maximum is $0.25$ for $p_1=0.25, t_1=1, p_2=0, t_2=4.53$. For this case, the toll-setter has to obtain all income from station $T_1$ as no traffic will take station $T_2$ due to the incentive on the non-tolled road, i.e., $p_2=0$. Like before, from equation \ref{eq:p2response}a, for the given values of $p_1$ and $D$, any $t_2\ge0$ is a solution, as this yields $p_2=0$ for the same objective.
\end{itemize}
\end{document}